\documentclass[a4paper,11pt]{article}
\pdfoutput=1

\usepackage{tcs}
\DeclareUnicodeCharacter{1EF3}{\'{y}}
\begin{document}

\title{
High-Temperature Gibbs States are Unentangled \\ and Efficiently Preparable}
\author{
Ainesh Bakshi \\
\texttt{ainesh@mit.edu} \\
MIT
\and
Allen Liu \\
\texttt{cliu568@mit.edu} \\
MIT
\and
Ankur Moitra \\
\texttt{moitra@mit.edu} \\
MIT
\and
Ewin Tang \\
\texttt{ewin@berkeley.edu} \\
UC Berkeley
}
\date{}

\maketitle

\begin{abstract}

We show that thermal states of local Hamiltonians are separable above a constant temperature. 
Specifically, for a local Hamiltonian $H$ on a graph with degree $\degree$, its Gibbs state at inverse temperature $\beta$, denoted by $\rho =e^{-\beta H}/ \tr(e^{-\beta H})$, is a classical distribution over product states for all $\beta < 1/(c\degree)$, where $c$ is a constant.
This proof of \emph{sudden death of thermal entanglement} resolves the fundamental question of whether many-body systems can exhibit entanglement at high temperature.

Moreover, we show that we can efficiently sample from the distribution over product states.
In particular, for any $\beta < 1/( c \degree^2)$, we can prepare a state $\eps$-close to $\rho$ in trace distance with a depth-one quantum circuit and $\poly(n, 1/\eps)$ classical overhead.
\footnote{In independent and concurrent work, Rouzé, França, and Alhambra~\cite{rfa24} obtain an efficient quantum algorithm for preparing high-temperature Gibbs states via a dissipative evolution.}

\end{abstract}

\thispagestyle{empty}
\clearpage
\newpage

\microtypesetup{protrusion=false}
\tableofcontents{}
\thispagestyle{empty}
\microtypesetup{protrusion=true}
\clearpage
\setcounter{page}{1}

\section{Introduction}

A central motivation behind the study of quantum many-body systems is to understand the behavior of entanglement, i.e.\ non-classical correlations.
Spurred by experimental breakthroughs in quantum simulation and quantum phases of matter, a rich body of work has sprung around understanding entanglement in Gibbs states, which model quantum systems at thermal equilibrium (see~\cite{alhambra22} and references therein).
The goal of this literature is to quantify the allowable entanglement of these states as dictated by the locality structure of their underlying Hamiltonians.

We ask perhaps the simplest and most fundamental question in this domain: 
\begin{quote}
\begin{center}
\emph{At what temperatures can many-body systems be entangled?}
\end{center}
\end{quote}

The study of the relationship between temperature and entanglement dates back to the start of modern quantum information.
However, the wealth of ideas generated since then has yielded surprisingly little insight on this question.
Computational investigations suggest that specific models are unentangled above a constant temperature \cite{abv01,gkvb01,amd15}, but these results only study entanglement measures like concurrence and negativity, which can evaluate to zero for entangled states.
As for rigorous results, a simple argument about the geometry of entanglement shows that Gibbs states are unentangled, i.e.\ separable, at some finite temperature, but this temperature is exponentially large in the system size and therefore unphysical for many-body systems~\cite{gb03}.

Recent research also sheds little light on this question.
Over the past two decades, an exciting new suite of techniques has been developed to prove results about correlations in Gibbs states~\cite{alhambra22}.
However, these results, including thermal area laws~\cite{wvhc08,kaa21}, bounds on conditional mutual information~\cite{kb19}, local indistinguishability~\cite{kgkre14}, efficient state preparation~\cite{bk18,bcglpr23}, and efficient learning algorithms~\cite{AAKS21,blmt24}, only bound entanglement through proxy correlation measures that combine both classical and quantum correlations.\footnote{
    In the referenced work, such measures include mathematical quantities like two-point correlation functions, mutual information and conditional mutual information, and locality of effective Hamiltonians, along with operational and algorithmic notions like rapid thermalization (or low-depth state preparation more broadly) and efficient learning from local observables.
}
Consequently, these results can only prove lack of entanglement at \emph{long range}, when classical correlations are sufficiently small.
This remains true even for results which assume the Gibbs state is above a critical temperature~\cite{bk18,hms20,hkt21}.
The existing literature reflects a gap in our set of tools: we only have exiguous
 methods for controlling entanglement independently of classical correlations.
Taken together, this body of work might even lead one to believe that quantum correlations, like classical correlations, can exist at any constant temperature.

In stark contrast to prior work, we show that above some constant temperature, the Gibbs state of any local Hamiltonian exhibits \emph{zero} entanglement.

\begin{theorem}[Informal version of \cref{thm:gibbs-separable}]
    \label{thm:separable-informal}
    Consider a system of $\qubits$ qudits on a constant-dimensional lattice, and a Hamiltonian $H = \sum_a H_a \in \C^{d^\qubits \times d^\qubits}$ where every term is local with respect to the lattice.
    Then there is a constant critical temperature such that, above this temperature, the Gibbs state $\rho = e^{-\beta H} / \tr(e^{-\beta H})$ is separable, i.e.\ expressible as a mixture over product states.
\end{theorem}

With this result, we witness a \emph{``sudden death of thermal entanglement''}: there is a constant critical temperature, above which correlations are purely classical.
Alternatively formulated, this suggests that an entangled state, if weakly coupled to a bath at high temperature and allowed to thermalize, will become fully unentangled at a \emph{finite} time.
To the best of our knowledge, this result is the first to show strong bounds on short-range entanglement at constant temperature.

Our structural result for high-temperature Gibbs states arises from an \emph{algorithmic} technique, which we apply to the computational task of preparing quantum Gibbs states.
This task, known as quantum Gibbs sampling, has been studied extensively~\cite[Table~1]{ckbg23}, dating back to the work of Temme, Osborne, Vollbrecht, Poulin, and Verstraete~\cite{tovpv11}.
However, remarkably little is known about when Gibbs states can be prepared efficiently.
Despite a wealth of approaches and proposals, efficient algorithms have only been rigorously established in fairly restricted settings, such as for Hamiltonians with constant operator norm~\cite{gslw18}, commuting Hamiltonians~\cite{kb16}, 1D Hamiltonians~\cite{bk18,bcglpr23}, or under strong assumptions, like the eigenstate thermalization hypothesis~\cite{cb21}.
On the other hand, sampling from the Gibbs distribution at low temperature is known to be NP-hard, even in the classical setting~\cite{sly10, ss14, gsv16}.
Thus, there is a natural target: Are all high-temperature Gibbs states efficiently preparable?

Our second main result resolves this question.

\begin{theorem}[Informal version of \cref{thm:main-gibbs}]
    \label{thm:preparable-informal}
    Consider a system of $\qubits$ qubits on a constant-dimensional lattice, and a Hamiltonian $H = \sum_a H_a \in \C^{2^\qubits \times 2^\qubits}$ where every term is local with respect to the lattice.
    Then there is a constant critical temperature such that, above this temperature, the Gibbs state $\rho = e^{-\beta H} / \tr(e^{-\beta H})$ can be prepared to $\eps$ error with a depth-1 quantum circuit and $\poly(\qubits, \log(1/\eps))$ classical overhead.
\end{theorem}

The task of preparing a Gibbs state is a natural place to look for quantum speedups.
However, since our algorithm is almost entirely classical, we can conclude that this task offers no super-polynomial quantum speedup for temperatures above a fixed constant.
On the other hand, at temperatures below a fixed constant, preparing a Gibbs state is NP-hard, so this task is also unsuitable for finding advantage, at least in the worst case.
Going forward, finer-grained control of the separability and computational thresholds across different classes of Hamiltonians appears crucial to understanding Gibbs sampling, both as a testbed for quantum thermalization and as a candidate for quantum advantage.

\subsection{Results}

We now formally state our results.
For this section, we state our results for systems on qubits $(d = 2)$.
We consider a class of Hamiltonians which are local with respect to an underlying graph, which we call \emph{low-intersection} Hamiltonians.

\begin{definition}[Hamiltonian] \label{def:hamiltonian}
    A \emph{Hamiltonian} on $\qubits$ qubits is an operator $H \in \mathbb{C}^{2^\qubits \times 2^\qubits}$ that we consider as a sum of $\terms$ local \emph{terms} $H_a$, with $H = \sum_{a=1}^\terms H_a$.
    We also refer to these qubits as \emph{sites}.
    For normalization, we assume that the terms have bounded operator norm, $\norm{H_a}_{\op} \leq 1$.
    We say this Hamiltonian is \emph{$\locality$-local} if every term $H_a$ is supported on at most $\locality$ qubits: $\abs{\supp(H_a)} \leq \locality$.

    We say the Hamiltonian has Pauli terms if, for every $a \in [\terms]$, $H_a = \lambda_a E_a$ for $E_a \in \locals = \braces{\id, \sigma_x, \sigma_y, \sigma_z}^{\otimes \qubits}$ a tensor product of Pauli matrices (see \cref{def:paulis}) and $-1 \leq \lambda_a \leq 1$.
\end{definition}

\begin{definition}[Low-intersection Hamiltonian]
    \label{def:low-intersection-ham}
    For an $\qubits$-qubit Hamiltonian $H = \sum_{a=1}^\terms H_a$, we define its underlying \emph{dual interaction graph} $\graph$ to have vertices labeled by $\braces{1,2,\dots,\terms}$ and an edge between $a$ and $b$ if and only if $\supp(H_a) \cap \supp(H_b) \neq \varnothing$.
    We say $H$ has degree $\degree$ if the degree of every vertex in $\graph$ is at most $\degree$.\footnote{
        In the classical literature, it is more conventional to define the degree of a spin system to be the degree of the \emph{interaction graph}, the most number of terms which act on the same site.
        To be consistent with the quantum information literature, we use dual interaction graph degree, and note that this degree is at most $\locality$ times the degree of the interaction graph.
            }
    We call a Hamiltonian $H$ a \emph{$(\degree, \locality)$-low-intersection Hamiltonian} if $H$ has locality $\locality$ and degree $\degree$.
\end{definition}

Low-intersection Hamiltonians generalize geometrically local Hamiltonians in low-dimensional spaces, which is the type of Hamiltonian often considered in physically motivated settings.
For example, a $\locality$-local Hamiltonian which is geometrically local with respect to a $D$-dimensional lattice is a $(\locality(2D)^{\locality-1}, \locality)$-low intersection Hamiltonian.

Our structural result states that low-intersection Gibbs states are separable at high temperature.

\begin{theorem}[High-temperature Gibbs states are separable]
    \label{thm:separable}
    Let $H$ be a $(\degree, \locality)$-low-intersection Hamiltonian on $\qubits$ qubits with Pauli terms.
    Then, for $ \beta < 1/(100\degree \locality) $, the corresponding Gibbs state $\rho = e^{-\beta H} /\tr\parens{ e^{-\beta H} }$ is separable.

    In particular, $\rho$ can be expressed as a distribution over stabilizer product states,\footnote{
        The single-qubit stabilizer states, written here as $\frac12(\id \pm \sigma_z)$, $\frac12(\id \pm \sigma_x)$, and $\frac12(\id \pm \sigma_y)$, may be more familiar as $\{\ketbra{0}{0}, \ketbra{1}{1}, \ketbra{+}{+}, \ketbra{-}{-}, \ketbra{\ii}{\ii}, \ketbra{-\ii}{-\ii}\}$.
    } i.e.\ $ A_1 \otimes A_2 \otimes \ldots \otimes A_n $ such that for each $j \in [n]$, $A_j \in \Set{ \tfrac12(I \pm \sigma_x), \tfrac12(I \pm \sigma_y),  \tfrac12(I \pm \sigma_z)}$.
\end{theorem}

Assuming that the terms are Pauli can be done without loss of generality: an arbitrary $H_a$ supported on $\locality$ qubits can be expanded into the Pauli basis, inflating the number of terms by at most a factor of $4^\locality$.
We also show an analogous result for qudits (\cref{thm:gibbs-separable}).

A simple argument shows that there is \emph{some} temperature for which Gibbs states are separable: the Gibbs state at infinite temperature ($\beta = 0$) is the maximally mixed state, $\id/d^\qubits$, which is in the interior of the convex hull of product states~\cite{bcjlps99}.
So, as $\beta$ tends to zero, $\rho$ will eventually enter the interior of this convex hull, making it separable.
However, this only implies separability for $\beta < \exp(-c n)$ for a fixed constant $c$, as Gurvits and Barnum show that this is the radius of a ``separable ball'' around the maximally mixed state~\cite{gb03}.
Such arguments do not use the locality structure of the Hamiltonian, and without using it, this $\beta$ is tight~\cite{as06}.
We show that Gibbs states are separable above a temperature that is independent of system size: this distinction is important, as for real-world materials, $n$ is very large and temperatures which depend on $n$ are non-physical.

Additionally, prior work on thermal area laws shows that the mutual information between two subsystems $L$ and $R$ of a Gibbs state is bounded by a linear function in the size of the surface area of $L$, $\beta \abs{\partial L}$~\cite{wvhc08,kaa21}.
Here, mutual information serves as a proxy for entanglement, mirroring area laws for entanglement entropy in ground states~\cite{hastings07,aklv13,aag22}.
On the other hand, our results treat entanglement directly and demonstrate that, for any temperature higher than a fixed constant, it is identically zero.

\begin{remark}[Entanglement death under an external field]
    Kuwahara and Hatano~\cite{kh11} show that for a two-qubit system with a weak interaction term, $H = \beta h_{12} + h_1 + h_2$, the corresponding thermal state $e^H/\tr(e^H)$ can still exhibit entanglement provided that the size of the one-local terms, $\norm{h_1 + h_2}_{\op}$, is correspondingly large, scaling as $\log(1/\beta)/\beta$.
    This shows that our assumption that $\abs{\lambda_a} \leq 1$ cannot be removed for one-local terms. It also shows that although increasing the external field may not enhance correlations, it influences the extent to which those correlations can be represented classically.
        \end{remark}

Next, we state our result on efficiently preparing Gibbs states. 

\begin{theorem}[High-temperature Gibbs states are efficiently preparable]
\label{thm:main-gibbs}
Given $\eps \in (0, 1)$, a $(\degree, \locality)$-low-intersection Hamiltonian $H$ on $\qubits$ qubits with Pauli terms, and a $\beta <  \beta_c = 1/(\const\degree \locality)$, where $\const$ is a fixed universal constant, consider the Gibbs state $\rho = e^{-\beta H }/ \tr\Paren{ e^{-\beta H} }$.
Then, there exists a classical randomized algorithm which runs in time
\[
    \bigOt*{\qubits^{ 6 +\frac{\log(\degree)}{\log(\beta_c/\beta)} } \cdot \log^3(1/\eps) \cdot  \poly(\locality, \degree)}.
\]
The algorithm outputs a classical description of a product state $\wh{\rho} = A_1 \otimes \dots \otimes A_n$ where every $A_j$ is an eigenvector of a Pauli matrix, $A_j \in \Set{ \tfrac12(I \pm \sigma_x), \tfrac12(I \pm \sigma_y),  \tfrac12(I \pm \sigma_z)}$,
such that the mixture over $\wh{\rho}$ is close to $\rho$ in trace distance,
\[
    \Norm{ \rho - \E[\wh{\rho}]}_{1 }  \leq \eps,
\]
where the expectation is only over the randomness of the algorithm.
\end{theorem}

\cref{thm:preparable-informal} follows from this theorem by considering $\beta < \beta_c / \degree$, so that the exponent on $\qubits$, $6 + \frac{\log(\degree)}{\log(\beta_c/\beta)} < 7$, is a constant.
The output $\wh{\rho}$ can be prepared on a quantum computer with one layer of single-qubit gates.

With this algorithm, one can prepare a copy of $\rho$ by running our randomized algorithm, taking the classical description of $\wh{\rho}$, and preparing it with a depth-one quantum circuit.  Note that preparing $\rho$ in expectation is equivalent to preparing $\rho$, since one can just ``forget the algorithm's steps'' to get a copy of $\rho$ without classical side correlations.
Even stronger, if one performs the algorithm coherently, it outputs a purification of $\rho$.

\paragraph{On temperature.}
Our bound on the critical temperature for the algorithm cannot be significantly improved.
Work of Sly and Sun~\cite{ss14} shows that approximately sampling from the anti-ferromagnetic Ising model on a $\Delta$-regular graph is NP-hard at the ``uniqueness threshold'' for the $\Delta$-regular tree, which is at $\beta = \bigTheta{1/\Delta} = \bigTheta{1/\degree}$~\cite{sst14}.
This hardness statement for classical Gibbs sampling implies hardness for the more general problem of quantum Gibbs sampling.

\paragraph{On quantum advantage.}
Our results show that any reasonable downstream computational task derived from high-temperature Gibbs sampling can be performed efficiently on a classical computer, and so is not a good candidate for quantum advantage.
In some sense, this strengthens prior work on algorithmic cluster expansion, which gives efficient classical algorithms for various tasks associated with high-temperature Gibbs states, including computing partition functions~\cite{mh21}, estimating local observables $\tr(O\rho)$, and sampling from $\rho$ measured in a product basis~\cite{yl23}.
Indeed, our structural result shows that such states are almost entirely classical, giving a simple explanation for why these efficient classical algorithms are possible.

\subsection{Related work}

\paragraph{Locality in high-temperature Gibbs states.}
Several works in the quantum information literature focus on high-temperature Gibbs states.
These bound correlation measures which do not distinguish quantum and classical, including covariance of observables~\cite{hms20} and local indistinguishability~\cite{kgkre14,bk18}.
Our work, which shows a lack of quantum correlation even in the presence of classical correlations, is a significant departure from this prior work.

\paragraph{Sudden death of entanglement.}
Sudden death of entanglement is the phenomenon that two entangled qubits, when subject to environmental noise, does not exhibit exponentially decaying entanglement with time, as classical correlations do, but rather become entirely disentangled after a finite amount of time~\cite{ye09}.
The body of literature on ESD (entanglement sudden death) focuses on analyzing two-qubit systems under various noise models~\cite{ye04,fmb05,aj08} and experimental demonstrations of ESD~\cite{amhswrd07}.
Its study as a phenomenon of many-body systems is more limited, likely because even defining a measure of entanglement for mixed states is non-trivial~\cite{hhhh09}, and separability is difficult to detect for large systems.

For many-body systems, prior work studies the sudden death of \emph{entanglement negativity}, an entanglement monotone which can be computed efficiently~\cite{vw02}, in specific spin systems, either through heuristic arguments or numerical calculations~\cite{abv01,gkvb01,amd15,sdhs16,hc18}.
States with zero entanglement negativity need not be separable~\cite{hhh98}.
So, our work is vastly more general and formally proves separability of Gibbs states at high temperature, a much stronger result than lack of entanglement negativity.

\paragraph{Other kinds of quantum systems.}
There are natural quantum systems which exhibit entanglement at even infinite temperature: a recently studied example is systems with non-abelian symmetries~\cite{lprs24,mlwhs24}.
From the perspective of statistical physics, our setting of considering states $\rho \propto e^{-\beta H}$ is closest to the grand canonical ensemble, where there is no particle nor energy conservation; when conservation is enforced, entanglement appears able to persist to high temperature, suggesting a possible difference between these statistical ensembles.

\paragraph{Geometry of entanglement.}
Prior arguments for the separability of Gibbs states use properties about the geometry of entanglement~\cite{bz06}, and in particular, the distance from the maximally mixed state to the closest non-separable state.
This distance is exponentially small in the number of subsystems $\qubits$, as shown by Gurvits and Barnum~\cite{gb03}, and subsequent work has further strengthened and generalized this result~\cite{as06,hildebrand07,wk23}.
As previously observed, this suffices to show separability for Gibbs states for exponentially small $\beta$~\cite{cr06,pw24}.
Since our result shows separability at constant temperature, it implies that sudden death of thermal entanglement occurs at physically reasonable temperatures, even for many-body systems.

\paragraph{Classical Gibbs sampling.}
Classical Gibbs sampling is, comparatively, well-understood and researchers have characterized sharp phase transitions wherein there is some critical temperature, above which sampling the Gibbs state is computationally efficient and below which it is computationally hard \cite{sly10, ss14, gsv16}.
The fact that there are such wide gaps in our understanding of quantum Gibbs sampling is especially surprising given the diverse range of techniques we have for classical spin systems, such as path coupling \cite{bd97}, canonical paths \cite{jsv04}, correlation decay \cite{weitz06}, abstract polymer models \cite{kp86}, zero-free regions \cite{barvinok16}, spectral independence \cite{alg21}, and stochastic localization \cite{ce22}.
Our results can be thought of as a sampling-to-counting reduction for quantum systems.
Thus, we give an algorithmic alternative to directly working with Lindbladians of dissipative evolutions.

\paragraph{Concurrent work on high-temperature Gibbs sampling.}
In concurrent and independent work, Rouzé, França, and Alhambra~\cite{rfa24} prove that the dissipative evolution studied by Chen, Kastoryano, and Gilyén~\cite{ckg23} has a constant spectral gap at high temperature, showing that this evolution is an efficient Gibbs sampling algorithm at high temperature.
The techniques are significantly different than ours, controlling the evolution by viewing it as a perturbation of an infinite-temperature dissipation.
We do not analyze such an evolution.

\paragraph{Cluster expansion and abstract polymer models.}
Our approach is based on cluster expansion, which allows quantities like the log-partition function and marginals of high-temperature Gibbs states to be expressed as exponentially decaying Taylor series.
This tool has been used to show a variety of efficient algorithms around Gibbs states and real-time evolution, including the computation of partition functions~\cite{mh21}, learning of high-temperature Gibbs states~\cite{hkt21}, and sampling from the measurement distribution of a high-temperature Gibbs state~\cite{yl23}.

Among these, the latter sampling result of Yin and Lucas bears the most similarity to our result, using a sampling-to-counting reduction with cluster expansion to give a classical algorithm to sample from the measurement probabilities of a Gibbs state in, say, the computational basis.
Our work also implies an efficient algorithm for this task, and achieves a stronger result by tackling the additional challenge of performing these arguments ``without measuring''.  

\section{Technical overview}

We now describe the key technical ingredients used to obtain \cref{thm:separable,thm:main-gibbs}.

To prove separability, we construct a randomized algorithm which outputs an un-normalized state $\sigma$ such that, over the randomness of the algorithm, $\E[\sigma] = e^{-\beta H}$.
For small enough $\beta$, $\sigma$ is separable, thereby proving that $e^{-\beta H}$ and $\rho$ are separable.
Our algorithm is a careful sampling-to-counting reduction: we start with all $\qubits$ sites ``unpinned''.
Then, we sample a separable state on one site according to the marginal of $e^{-\beta H}$, and are left with a posterior-like state on the rest of the sites.
We then iterate, sampling a neighboring site according to the marginal of the posterior state, and so on until we have ``pinned'' all of the sites, at which point we output the resulting state as $\sigma$.
Notice that the idea of pinning site by site only makes sense if $e^{-\beta H}$ is indeed separable.

The main technical challenge is controlling the effect of the pinned sites on the un-pinned sites: if these correlations grow too large, they can push us out of the hull of separable states, which in the algorithm corresponds to the probabilities in the distribution becoming negative.
We show that this does not occur via a potential argument.
Pinned sites only affect nearby sites, so for a choice of $\beta$ sufficiently small in terms of degree, over the course of the algorithm, sites stay close to maximally mixed, and therefore we can pin them without issue.

To turn the proof of separability into a state preparation algorithm, we need to account for normalization: for the separability argument, we were not concerned with mild, constant-factor errors in sampling probabilities, provided that these did not lead to a failure in separability.
In particular, since we need to output a quantum state, $\sigma / \tr(\sigma)$, to get $\rho = e^{-\beta H} / \tr(e^{-\beta H})$ in expectation, the probability the algorithm outputs $\sigma$ is off by a factor of $\tr(\sigma) / \tr(e^{-\beta H})$.
This can be fixed by adjusting the sampling probabilities for pinning by an appropriate factor which corresponds to an expectation value of a local observable with respect to $\rho$; this is the ``counting'' of the sampling-to-counting reduction.
Doing this directly achieves a sub-optimal running time, so we instead use an alternative reduction goes from sampling to weak approximate counting, inspired by a similar algorithm of Sinclair and Jerrum~\cite{sj89}.

Throughout this section, we use the following notation.
We take $H = \sum_{a = 1}^\terms H_a$ to be a $(\degree, \locality)$-low-intersection Hamiltonian, and let $\beta< 1/(100 \degree \locality)$.
We also consider this Hamiltonian in various restricted forms: for a set $S \subseteq [\qubits]$, $H_{(S)}$ is the sum over terms $H_a$ whose support $\supp(H_a)$ intersects $S$; and $H^{(S)}$ is the sum over terms $H_a$ whose support is in $S$ (see \cref{def:restricted-ham}).
So, for example, $H_{(\braces{j})}$ is the sum over terms which contain $j$ in their support, and we abbreviate this as $H_{(j)}$; and $H - H_{(j)} = H^{([n] \setminus \braces{j})}$.

\subsection{Gibbs states are unentangled} \label{subsec:tech-unentangled}

The separability argument proceeds as follows.
To show that $e^{-\beta H}$ is separable, we consider writing it as
\begin{align} \label{eq:first-decomp}
    e^{-\beta H}
    = e^{-\frac{\beta}{2}H^{([2, \qubits])}} \underbrace{\parens[\big]{e^{\frac{\beta}{2}H^{([2, \qubits])}} e^{-\beta H} e^{\frac{\beta}{2}H^{([2,\qubits])}}}}_{\text{\eqref{eq:first-decomp}.(1)}} e^{-\frac{\beta}{2}H^{([2,\qubits])}},
\end{align}
where $[2, \qubits] = [\qubits] \setminus \braces{1}$.
This separates the exponential into the outside terms, which have no dependence on site 1, and \eqref{eq:first-decomp}.(1), which models the ``correction'' needed to go from the Gibbs state on $[2, \qubits]$ to the full Gibbs state.
We will see that the expression \eqref{eq:first-decomp}.(1) is separable.
In particular, this is enough to show that $e^{-\beta H}$ is separable between site 1 and the rest of the sites: because \eqref{eq:first-decomp}.(1) is separable, we can write it as a sum of positive semi-definite product operators $\sum_a X_a = \sum_a X_{a,1} \otimes X_{a,[2,\qubits]}$, and then conclude\footnote{
    Here, we overload notation and use $e^{-\frac{\beta}{2}H^{([2,\qubits])}}$ to denote both the $2^\qubits \times 2^\qubits$ operator that is $\id$ on the first site, and the $2^{\qubits-1} \times 2^{\qubits-1}$ operator with the first site factored out.
}
\begin{align}
    e^{-\beta H}
    &= e^{-\frac{\beta}{2}H^{([2, \qubits])}} \parens[\Big]{\sum_a X_{a,1} \otimes X_{a,[2,\qubits]}} e^{-\frac{\beta}{2}H^{([2,\qubits])}} \nonumber \\
    &= \sum_a X_{a,1} \otimes \parens[\big]{e^{-\frac{\beta}{2}H^{([2, \qubits])}} X_{a,[2,\qubits]} e^{-\frac{\beta}{2}H^{([2,\qubits])}}} \,. \label{eq:sep-one-site}
\end{align}
We say here that we have ``pinned'' site 1: we have written $e^{-\beta H}$ as a convex combination of operators where the state on site 1 is fixed and separate from the Gibbs-like state on the remaining sites.
This suggests a path to iterating this argument: within this sum, separate another site $j$ from $[2, \qubits]$, and recurse, continuing to argue that the intermediate operators are separable.
For example, to do this for site $j = 2$, we would write the expression in the parenthetical as
\begin{align} \label{eq:next-site}
    e^{-\frac{\beta}{2}H^{([2, \qubits])}} X_{a,[2,\qubits]} e^{-\frac{\beta}{2}H^{([2,\qubits])}}
    = e^{-\frac{\beta}{2}H^{([3, \qubits])}} \underbrace{\parens[\big]{e^{\frac{\beta}{2}H^{([3, \qubits])}} e^{-\frac{\beta}{2}H^{([2, \qubits])}} X_{a,[2,\qubits]} e^{-\frac{\beta}{2}H^{([2, \qubits])}} e^{\frac{\beta}{2}H^{([3, \qubits])}}}}_{\text{\eqref{eq:next-site}.(1)}} e^{-\frac{\beta}{2}H^{([3,\qubits])}},
\end{align}
and argue that \eqref{eq:next-site}.(1) is separable.

To execute this argument, there are two considerations to handle.
First, we need to understand how to choose an appropriate set of $X_a$'s; the decomposition of a separable operator into product operators is not unique, and there are some such operators for which \eqref{eq:next-site}.(1) is not separable.
Our argument gives a decomposition where $X_{a,[2,\qubits]} \propto \id + c_a P_a$ where $P_a$ is a tensor product of Paulis and $c_a$ is exponentially small in the support size of $P_a$, $\abs{c_a} \leq (1-\frac{3}{5\locality})^{\abs{\supp(P_a)}}$.
Second, we need to decide how to pin sites in the recursion.
The naive approach is to do this non-adaptively, which would effectively amount to analyzing the decomposition
\begin{align*}
    e^{-\beta H} =
    \parens[\Big]{\prod_{i=1}^\qubits (e^{-\frac{\beta}{2}H^{([i, \qubits])}}e^{\frac{\beta}{2}H^{([i+1, \qubits])}})}^\dagger
    \parens[\Big]{\prod_{i=1}^\qubits (e^{-\frac{\beta}{2}H^{([i, \qubits])}}e^{\frac{\beta}{2}H^{([i+1, \qubits])}})}
\end{align*}
from the inside out.
It turns out to be natural for us to choose the next site to pin, $j$, based on the $X_a$; in particular, we choose $j$ to be in the support of $X_a$.
Essentially, every time we pin a site, it can increase the size of the correlation, $\abs{c_a}$; but provided we remove a site involved in the correlation each time, the correlation never grows large enough to generate entanglement.

\paragraph{Intuition for separability.}
Understanding when an operator is separable is difficult in general~\cite{hhhh09}.
Our proof works with a sufficient condition for separability: that the operator is sufficiently close to the identity.
The identity operator $\id$ is separable because $\id = \id_1 \otimes \dots \otimes \id_\qubits$.
Moreover, perturbations of the identity operator by a tensor product of Paulis $P = P_1 \otimes \dots \otimes P_\qubits$ are also separable: $\id + P$ is separable because it is a sum over $\ketbra{\psi}{\psi}$, where $\ket{\psi} = \ket{\psi_1}\otimes \dots \otimes \ket{\psi_\qubits}$ is an eigenvector of $P$ with eigenvalue $+1$.
Because these eigenvectors are product, $\id + P$ is separable.

We can use this to build more general sufficient conditions for separability:\footnote{
    In fact, all these conditions suffice for writing the operator as a positive sum of (density matrices of) stabilizer product states, not just PSD product operators.
} for any (Hermitian) operator $Y$, $\id + Y$ is separable provided the perturbation is exponentially small in support size, $\norm{Y}_{\op} \leq 0.25^{\abs{\supp(Y)}}$.
To see this, write $Y = \sum_{P}\frac{1}{2^\qubits}\tr(YP) P$ where the sum is over all Paulis $P$ such that $\supp(P) \subseteq \supp(Y)$.
This sum has $4^{\abs{\supp(Y)}}$ summands.
Then, we can write $\id + Y = \sum_P \frac{1}{4^{\abs{\supp(Y)}}}(\id + \frac{4^{\abs{\supp(Y)}}}{2^\qubits}\tr(YP) P)$, where $\frac{4^{\abs{\supp(Y)}}}{2^\qubits} \tr(YP) \leq 4^{\abs{\supp(Y)}}\norm{Y}_{\op}$.
$\id + c P$ is separable provided $\abs{c} \leq 1$; so, $\id + Y$ is separable provided $\norm{Y}_{\op} \leq 4^{-\abs{\supp(Y)}}$.

A generalization gives us our main piece of intuition: \emph{an operator is separable if it is a quasi-local perturbation of the identity}.
Formally, an operator $X$ is separable if we can write it as
\begin{align} \label{eq:quasilocal-is-sep}
    X = \id + \sum_{t=1}^\infty \sum_{b \in \mathcal{I}^{(t)}} Y_b
\end{align}
where $\abs{\supp(Y_b)} \leq t$ for every $b \in \mathcal{I}^{(t)}$; and for each $t = 1,2,\dots$, the operators at that level have bounded norm, $\nu_t \coloneqq \sum_{b \in \mathcal{I}^{(t)}} \norm{Y_b}_{\op} \leq 0.1^{t}$.
The argument is similar, but we phrase it in an algorithmic way, like the separability proof does.
Consider sampling a random index $t \in \braces{1,2,\dots}$ with probability $2^{-t}$, and then sampling a random term $b \in \mathcal{I}^{(t)}$ with probability $\frac{\norm{Y_b}_{\op}}{\nu_t}$.
Then the random matrix $\id + 2^t\nu_t \frac{Y_b}{\norm{Y_b}_{\op}}$ is equal to $X$ in expectation, and is separable, because the perturbation has sufficiently small operator norm.
So, $X$ is a convex combination of separable operators, and thus is separable.

\paragraph{Sampling a term.}
The expression from \eqref{eq:first-decomp}.(1) is separable for the same reason: this operator, which models the correlation between one site and the rest of the sites, is a quasi-local perturbation of identity.
To understand this, we consider the Taylor expansion of $e^{-\beta H} \cdot e^{\beta\parens{ H- H_{(j)} }}$.
For our choice of $\beta$, this series converges exponentially fast and we can sample a single term from it in an unbiased manner.

Expanding the Taylor series, we show that we can write
\begin{equation}
\label{eqn:series-expansion}
    e^{-\beta H} \cdot e^{\beta\parens{ H-H_{(j)} }  } = \sum_{t=0}^{\infty} \frac{\beta^t}{t!} f_t\parens{ H, H_{(j)}}, 
\end{equation}
where $f_t\parens{ H, H_{(j)}}$ is a matrix-valued, degree-$t$ polynomial satisfying the following recurrence relation:
\begin{align}
    f_{0}\parens{ H, H_{(j)}} &= \id \nonumber \\
    f_{t+1}\parens{ H, H_{(j)}} &= -\bracks{H, f_t\parens{ H, H_{(j)}} } - f_t\parens{ H, H_{(j)}}  H_{(j)},
\label{eqn:recursive-expansion}
\end{align}
From the structure of the recurrence, we prove inductively that we can write $f_t\parens{ H, H_{(j)}} = \sum_{b \in \mathcal{Q}^{(t)}} \mu_{b} H_{b_1}\dots H_{b_t}$, where $\mathcal{Q}^{(t)}$ are the tuples of terms $b = (b_1,\dots,b_t) \in [\terms]^t$ such that the corresponding support, $\supp(H_{b_1}) \cup \dots \cup \supp(H_{b_t})$, is connected and contains $j$.
So, $f_t\parens{ H, H_{(j)}}$ is $t\locality$-local.
Further, 
\begin{align} \label{eq:tech-mu-bound}
    \sum_{b \in \mathcal{Q}^{(t)}} \abs{\mu_b} \leq (3(\degree+1))^t t! \,.
\end{align}
So, for our choice of $\beta$, the series in \eqref{eqn:series-expansion} decays exponentially with $t$.
This is proven in \cref{thm:low-deg-approx}.
So, the full series expansion of \eqref{eqn:series-expansion} is
\[
    e^{-\beta H} \cdot e^{\beta\parens{ H-H_{(j)} }} =  I + \sum_{t = 1}^{\infty} \sum_{b \in \mathcal{Q}^{(t)}} \frac{\beta^t}{t!}\mu_{(b_1,\dots,b_t)} H_{b_1} \dots H_{b_t}.
\]
This is promising with regards to eventually proving separability, as this form is similar to that of \eqref{eq:quasilocal-is-sep}.
As with that series, we can write this series as an expectation of simple operators close to the identity.
Because the series decays exponentially, we can do this: let $\id + c E$ be a random variable, where, for any $b$ appearing in the series, $c = \frac{\beta^t}{t!} \cdot (2^t (3(\degree + 1))^t t!)$ and $E = \frac{\mu_b}{\abs{\mu_b}} H_{b_1}\dots H_{b_t}$ with probability $\frac{1}{2^t}\frac{\abs{\mu_{b}}}{(3(\degree+1))^t t!}$ (and $cE = 0$ with whatever probability mass remains).
By \eqref{eq:tech-mu-bound}, this probability distribution is well defined, and by construction, $\E[\id + cE] = e^{-\beta H} \cdot e^{\beta\parens{ H-H_{(j)} }}$.
Further, as desired, the operator norm of $cE$ is exponentially small in its support size.

\paragraph{Pinning the first site.} Next, we show that we can pin the Gibbs state at a site, say site $1$, which means we fix the $2\times 2$ density matrix on this site. We begin by decomposing the Gibbs state as in \eqref{eq:first-decomp}:
\begin{equation}
\label{eqn:decomposing-gibbs-state}
    e^{-\beta H}
    = e^{-\frac{\beta}{2}H^{([2, \qubits])}}
    \underbrace{ \parens[\big]{e^{-\frac{\beta}{2} H } \cdot e^{ \frac{\beta}{2}H^{([2,\qubits])}} }^\dagger}_{\eqref{eqn:decomposing-gibbs-state}.(1) }
    \cdot \underbrace{ \parens[\big]{e^{-\frac{\beta}{2} H } \cdot e^{ \frac{\beta}{2}H^{([2,\qubits])}} }}_{\eqref{eqn:decomposing-gibbs-state}.(2) }
    e^{-\frac{\beta}{2}H^{([2,\qubits])}}.
\end{equation}
Since $H^{([2,\qubits])} = H - H_{(1)}$, the two terms \eqref{eqn:decomposing-gibbs-state}.(1) and \eqref{eqn:decomposing-gibbs-state}.(2) are operators in the form appearing in \eqref{eqn:series-expansion}, with $\beta$ replaced by $\beta/2$. Therefore, we can invoke the aforementioned sampling primitive twice to obtain independent, unbiased samples $I + c_1 E_1$ and $I +  c_2 E_2$ such that, by linearity of expectation,
\begin{align*}
    &\expecf{}{ \Paren{ I + c_1 E_1  }^{\dagger} \cdot \Paren{ I + c_2 E_2 } }
    = e^{ \frac{\beta}{2}H^{([2,\qubits])}} e^{-\beta H} e^{ \frac{\beta}{2}H^{([2,\qubits])}},
\end{align*}
where the expectation is over the randomness of the sampling.
We would like to factor out site 1 from the rest of the sites for this decomposition, but $\Paren{I + c_1 E_1}^{\dagger}\Paren{I + c_2 E_2}$ is not Hermitian.
To make it Hermitian, we average it with its Hermitian conjugate, and then write it as a sum over simpler expressions:
\begin{multline}\label{eq:cases1}
    \frac{1}{2} \left( \Paren{ I +   c_1 E_1 }^\dagger  \Paren{ I +   c_2 E_2 }  + \Paren{ I +   c_2 E_2 }^\dagger  \Paren{ I +   c_1 E_1 } \right)
    \\
    = I + c_1 \parens[\Big]{\frac{E_1 + E_1^\dagger}{2}}  + c_2 \parens[\Big]{\frac{E_2 + E_2^\dagger}{2}}  + c_1c_2\parens[\Big]{\frac{E_1^\dagger E_2 + E_2^\dagger E_1}{2}}. \\
    = \frac{1}{3}\parens[\Bigg]{
        \underbrace{I + 3c_1 \parens[\Big]{\frac{E_1 + E_1^\dagger}{2}}}_{\text{\eqref{eq:cases1}.(1)}}
        + \underbrace{I + 3c_2 \parens[\Big]{\frac{E_2 + E_2^\dagger}{2}}}_{\text{\eqref{eq:cases1}.(2)}}
        + \underbrace{I + 3c_1c_2 \parens[\Big]{\frac{E_1^\dagger E_2 + E_2^\dagger E_1}{2}}}_{\text{\eqref{eq:cases1}.(3)}}
    }
\end{multline}
So, consider sampling $c_1E_1$ and $c_2E_2$ as before, and with probability $\frac13$ each, outputting $\id + cA$ to equal either \eqref{eq:cases1}.(1), \eqref{eq:cases1}.(2), or \eqref{eq:cases1}.(3).
Then the expectation is still the same, $\E[\id + cA] = e^{ \frac{\beta}{2}H^{([2,\qubits])}} e^{-\beta H} e^{ \frac{\beta}{2}H^{([2,\qubits])}}$, and $A$ is simple---in fact, since it comes from a product of Hamiltonian terms, it is Pauli.
Further, from this, we can conclude separability: $c_1E_1$ and $c_2E_2$ are both exponentially small in their support size, so $cA$ has this property too, and so $\id + cA$ is always separable.\footnote{
    The operator norm of the perturbation is a factor of three larger for $cA$ versus, say, $c_1E_1$, but this can be accounted for by decreasing $\beta$ by a corresponding factor.
    Technically, we do not need that $c$ is exponentially decaying in support size, since $A$ is always a Pauli; however, this property is used later in the argument.
}
We can use this to pin the first site and recurse on the subsystem of $\qubits-1$ qubits, as shown in \eqref{eq:sep-one-site}.\footnote{
    For example, we could write $\id + c(A_1 \otimes A_{[2,\qubits]}) = \frac12 (\id + A_1) \otimes (\id + cA_{[2,\qubits]}) + \frac12 (\id - A_1) \otimes (\id - cA_{[2,\qubits]})$, and these operators on the right-hand side, across all choices of $\id + cA$, form the $X_a$'s from \eqref{eq:sep-one-site}.
}
For ease of notation, though, we will defer explicitly pinning the sites to the end of the argument, and continue thinking about the full, $\qubits$-qubit system.

\paragraph{Recursing on the remaining sites.}
So, at this point, we have a randomly chosen matrix $\id + cA$ such that
\begin{equation}\label{eq:unbiased-step1}
    \expecf{}{ e^{-\frac{\beta}{2}H^{([2,\qubits])}} (\id + cA)  e^{-\frac{\beta}{2}H^{([2,\qubits])}} } = e^{-\beta H}. 
\end{equation}
Now, we want to peel off another site, following \eqref{eq:next-site}.
Consider the expression: 
\begin{multline}
    e^{-\frac{\beta}{2}H^{([2,\qubits])}} (\id + cA)  e^{-\frac{\beta}{2}H^{([2,\qubits])}} \\
    = e^{-\frac{\beta}{2}H^{([3, \qubits])}} \underbrace{\parens[\big]{e^{\frac{\beta}{2}H^{([3, \qubits])}} e^{-\frac{\beta}{2}H^{([2, \qubits])}} (\id + cA) e^{-\frac{\beta}{2}H^{([2, \qubits])}} e^{\frac{\beta}{2}H^{([3, \qubits])}}}}_{\text{\eqref{eqn:expanding-for-second-site}.(1)}} e^{-\frac{\beta}{2}H^{([3,\qubits])}}.
    \label{eqn:expanding-for-second-site}
\end{multline}
We would like to separate site 2 from $[3, \qubits]$ in \eqref{eqn:expanding-for-second-site}.(1) by decomposing it into simple product operators, as we did with site 1.
As before, we can obtain independent, unbiased samples of $e^{-\frac{\beta}{2}H^{([2, \qubits])}} e^{\frac{\beta}{2}H^{([3, \qubits])}}$, since this takes the form of the expansion from \eqref{eqn:series-expansion}.
Let $I + c_1' E_1'$ and $I + c_2' E_2'$ be the resulting samples, and thus
\begin{equation}\label{eq:cases2}
    \begin{split}
        & e^{-\frac{\beta}{2}H^{[3,\qubits]}} \Paren{ \frac{1}{2} ( I + c_1' E_1')^{\dagger}  (I + cA)   (I + c_2' E_2') + \frac{1}{2} ( I + c_2' E_2' )^{\dagger}  (I + cA) (I + c_1' E_1' ) } 
        e^{ -\frac{\beta}{2} H^{([3,\qubits])} } 
    \end{split}
\end{equation}
is an unbiased estimator of $e^{-\frac{\beta}{2}H^{([2,\qubits])}} (\id + cA)  e^{-\frac{\beta}{2}H^{([2,\qubits])}} $.
As before, we then expand these terms into simple, Hermitian operators, show that we can uniformly sample one of them to pin the second site, and then recurse.
These simple operators will either take the form of $\id + c'A'$ or $(\id + cA) \otimes (\id + c'A')$.
The second case occurs when, say, $cA$ is only supported on site 1, and so that perturbation can be factored out of the rest of the argument.
We call the class of simple operators we see in our algorithm \emph{configurations} (see \cref{def:configuration}).
For this configuration, which we call $\sigma$, we have that, over the randomness of pinning sites 1 and 2,
\begin{align*}
    \expecf{}{e^{-\frac{\beta}{2}H^{[3,\qubits]}} \cdot \sigma \cdot e^{-\frac{\beta}{2}H^{[3,\qubits]}}}
    = \expecf{}{e^{-\frac{\beta}{2}H^{([2,\qubits])}} \cdot (\id + cA) \cdot e^{-\frac{\beta}{2}H^{([2,\qubits])}}}
    = e^{ -\beta H }
\end{align*}

\paragraph{The full proof of separability.}
The general argument proceeds as follows.
We start with the configuration state $\sigma = \id$, and a set of unpinned sites $S$, which is initialized to $[\qubits]$.
Then, we pin one site at a time, removing sites from $S$ and modifying our configuration state $\sigma$ such that the following invariant always holds:
\begin{equation} \label{eq:intro-invariant}
    \expecf{}{e^{-\frac{\beta}{2} H^{(S)}} \cdot \sigma \cdot e^{-\frac{\beta}{2} H^{(S)}}} = e^{-\beta H},
\end{equation}
where the randomness is over the entire algorithm.
After extracting all sites, we are left with a random configuration $\sigma$, which looks like a tensor product of operators of the form $\id + cA$ where $A$ is Pauli, and the guarantee of the invariant, which states for $S = \varnothing$ that $\E[\sigma] = e^{-\beta H}$.
So, if we can argue that the coefficients are always bounded, $\abs{c} \leq 1$, then $\sigma$ is always separable, and thus $e^{-\beta H}$ is separable.

We show that $\abs{c} \leq 1$ can be maintained throughout the algorithm, so that $\sigma$ is always separable.
To do this, it suffices to argue that the $I + cA$ in e.g.\ \eqref{eq:unbiased-step1} satisfies $\abs{c} \leq 1$.
This perturbation is challenging to handle, since it changes, and can even increase, as we recurse.
To control the perturbation, we introduce a carefully chosen potential function: we prove inductively that
\[
    \abs{c} \leq \parens[\Big]{1 - \frac{3}{5\locality}}^{\abs{\supp(S) \cap \supp(A)}}
\]
throughout the algorithm (see \cref{lem:coeff-potential}).  In other words, we show that in each iteration of the recursion, if the support of $A$ on the unpinned sites grows, then the coefficient in front gets exponentially smaller, whereas if the coefficient grows, then the unpinned support of $A$ shrinks such that the coefficient does not have enough iterations to exceed 1.
Carrying out this argument requires carefully adjusting our sampling probabilities for the different cases in \eqref{eq:cases2} instead of sampling uniformly at random---see \cref{algo:single-step-2} for more details.  
Importantly, this argument also requires choosing the next site to pin based on the support of $A$: otherwise, the coefficient could grow without the unpinned support of $A$ shrinking.

\subsection{Gibbs states are efficiently preparable}

We now switch gears and describe our algorithm for preparing a Gibbs state, which is based on our algorithm for separability.

\paragraph{Implementing the separability algorithm efficiently.}
Recall that the algorithm in our proof of separability takes as input a description of a Hamiltonian $H$, along with a sufficiently small temperature $\beta$, and outputs a product operator $\sigma$ such that $\E[\sigma] = e^{-\beta H}$.
This algorithm can straightforwardly be made to run in $\bigO{\qubits \poly(\log(\qubits), \degree, \locality)}$ time (see \cref{lem:one-step-runtime}).

The primary complication is how to sample from the Taylor series in \eqref{eqn:series-expansion},
\begin{equation}
\label{eqn:series-truncation}
    e^{-\beta H} \cdot e^{\beta\parens{ H - H_{(j)} }  } = I +  \sum_{t=1}^{\infty} \frac{\beta^t}{t!} f_t\parens{ H, H_{(j)}} .
\end{equation}
The approach we described previously requires explicitly computing $f_t$ for whatever $t$ is sampled, paying exponential time in $t$ (which is constant in expectation, but is expected to be $\log(\qubits)$ at some point over the course of the algorithm).
This requires polynomial time; but we can design a faster sampling algorithm by exploiting the recursive structure of $f_t$, incurring only a linear dependence in $t$ (and therefore being constant time in expectation).
We do as follows: let $E_0 = I$ and suppose for some $t$, the current sample is $E_t$. Then, using \eqref{eqn:recursive-expansion}, 
\begin{equation*}
    f_{t+1}\paren{H , H_{(1)}}  =  \expecf{}{ - [H, E_t ]  - E_t \cdot H_{(1)} }.
\end{equation*}
The expression inside the expectation is a sum of $\bigO{t\degree}$ Pauli terms; we can sample each of these with an appropriate probability to get an unbiased sample.

\paragraph{Accounting for normalization.}
So, we can efficiently compute a product operator $\sigma$ whose expectation is $e^{-\beta H}$; in fact, $\sigma$ can be made to be a stabilizer product state, up to normalization.
However, if we then prepare $\sigma / \tr(\sigma)$ on a quantum computer, this does not give us $\rho = e^{-\beta H} / \tr(e^{-\beta H})$ in expectation, since the possible output $\sigma$'s will have different traces; essentially, we sample them with a probability which does not correspond to its ``true'' weight in $\rho$.

Let's return to our high-level description of the separability algorithm at the start of \cref{subsec:tech-unentangled}.
To pin our first site, we decompose
\begin{align*}
    e^{-\beta H} = \sum_a e^{-\frac{\beta}{2}H^{([2,\qubits])}} X_a e^{-\frac{\beta}{2}H^{([2,\qubits])}}
\end{align*}
where $X_a$ is a simple separable operator.
In fact, it will essentially be a local perturbation of the identity, $\id + cA$.
If we just care about separability, we can simply recurse our argument on the summands without concern for normalization.
However, since we wish to prepare the final output, we would like to maintain that our intermediate operators are genuine quantum states, e.g.\ unit trace.
So, we write
\begin{align} \label{eq:exact-probabilities}
    \rho = \frac{e^{-\beta H}}{\tr(e^{-\beta H})} = \sum_a \underbrace{\frac{\tr(e^{-\frac{\beta}{2}H^{([2,\qubits])}} X_a e^{-\frac{\beta}{2}H^{([2,\qubits])}})}{\tr(e^{-\beta H})}}_{\text{\eqref{eq:exact-probabilities}.(1)}}
    \underbrace{\frac{e^{-\frac{\beta}{2}H^{([2,\qubits])}} X_a e^{-\frac{\beta}{2}H^{([2,\qubits])}}}{\tr(e^{-\frac{\beta}{2}H^{([2,\qubits])}} X_a e^{-\frac{\beta}{2}H^{([2,\qubits])}})}}_{\text{\eqref{eq:exact-probabilities}.(2)}},
\end{align}
and now see that we want to sample the state \eqref{eq:exact-probabilities}.(2) with precisely the probability in \eqref{eq:exact-probabilities}.(1).
If we can do this, then we have successfully reduced the task of preparing $\rho$ to the task of preparing the simpler state \eqref{eq:exact-probabilities}.(2).
Similarly to marginal probabilities of classical Gibbs states, at high temperature, e.g.\ $\beta < 1/(100\degree)$, we can leverage cluster expansion techniques~\cite{kp86,hkt24,mh21} to estimate these probabilities efficiently (see~\cref{thm:estimation-log-partition,appendix-cluster-exp}).
This direct approach for sampling should work; however, estimating to $\eps$ multiplicative error using these methods incurs a polynomial dependence on $1/\eps$, which we would like to avoid.

\paragraph{Fast Gibbs state preparation.}
We show that it suffices to have $\bigO{1}$-error estimators to the partition function of the Gibbs state $\tr\parens{ e^{-\beta H}}$ and any restricted Gibbs state $\tr\parens{ e^{-\beta H^{(S)}}}$. 
Our reduction from $\eps$-approximate sampling to this ``weak'' constant-factor approximate counting is reminiscent of the one for self-reducible problems by Jerrum and Sinclair~\cite{sj89} and works as follows: we imagine setting up a tree over the space of possible choices that our sampling algorithm makes when producing an unbiased estimator to $e^{-\beta H}$.
These choices correspond to picking a site, expanding a series similar to \eqref{eqn:series-truncation}, and sampling one term in the expansion.
The tree $T$ has depth at most $\qubits$ and every node is labeled by the current configuration $\sigma$ and set of unpinned sites, $S$:
\begin{itemize}
    \item The root node is labeled by $\id$ and $[\qubits]$;
    \item A node at depth $k$ is indexed by a configuration $\sigma$ and some set $S \subseteq [\qubits]$ of size at most $\qubits - k$; implicitly, this represents the state $e^{-\frac{\beta}{2} H^{(S)}} \cdot \sigma \cdot e^{-\frac{\beta}{2} H^{(S)}}$, where the sites not in $S$ are pinned and the sites in $S$ are Gibbs-like;
    \item The possible children of a node are obtained by pinning (at least) one additional site in $S$ and modifying the configuration as described previously.
\end{itemize}
Ideally, we want to sample a product state by simply walking straight down this tree, going from a parent to the child $\sigma, S$ with probability proportional to $\tr(e^{-\frac{\beta}{2} H^{(S)}} \cdot \sigma \cdot e^{-\frac{\beta}{2} H^{(S)}})$.
However, this requires exactly computing these probabilities, as in \eqref{eq:exact-probabilities}.
Since we can only do approximate counting, we don't know the exact probabilities we should go down each of the branches in the tree.
Instead, we set up a random walk on this tree that goes both up and down (i.e. we could pin some site and then later unpin it and resample it) but has the desired stationary distribution on the leaves.
In particular, with respect to this distribution, the average of the product states at the leaves is close to the Gibbs state.

We show that this random walk mixes quickly: crucially, the separability algorithm is already doing a fairly good job at sampling with the desired probabilities, matching it to constant multiplicative error.
For example, at depth 1, when the first site is pinned, $\sigma$ always has the form of $\id + cA$ where\footnote{
    For the separability argument, we only needed that $\abs{c} \leq 1$.
    Here, we use a slightly lower $\beta_c$ so that $\abs{c} \leq 1/2$ holds instead.
} $\norm{cA}_{\op} \leq \frac12$, so $\tr(e^{-\frac{\beta}{2} H^{(S)}} \cdot \sigma \cdot e^{-\frac{\beta}{2} H^{(S)}})$ is always within a constant multiplicative factor of the constant $\tr(e^{-\beta H^{(S)}})$.
Because the walk mixes quickly, we can simply run the walk until it hits a leaf and output the product state corresponding to that leaf.

\section{Background}
\label{sec:background}

We begin by defining some notation.
We use the notation $[k] = \braces{1, 2, \dots, k}$ and take $\ii = \sqrt{-1}$ to be the imaginary unit.
On a graph $G = (V, E)$, we use ``neighborhood'' and ``neighbor'' in the closed sense: a node $v \in V$ is its own neighbor, and the neighborhood of $S \subseteq V$ is the union of the neighborhoods of all $v \in S$.
Throughout, $\bigO{f}$, $\bigOmega{f}$, and $\bigTheta{f}$ are big O notation, and $\bigOt{f} = \bigO{f \polylog(f)}$.

\subsection{Linear algebra}

We work in the Hilbert space corresponding to a system of $\qubits$ qudits with local dimension $d$, $\mathbb{C}^d \otimes \dots \otimes \mathbb{C}^d$.
For a subsystem $S \subseteq [\qubits]$ and an operator $A \in \C^{d^{\abs{S}} \times d^{\abs{S}}}$, we occasionally use the notation $A_S$ to refer to the operator which is $A$ on the subsystem $S$, and the identity otherwise.

For a matrix $A$, we use $A^\dagger$ to denote its conjugate transpose, $\norm{A}_\op$ to denote its operator norm, and $\norm{A}_1$ to denote its trace norm; for a vector $v$, we use $\norm{v}$ to denote its Euclidean norm.
For operators $A, B$, $A \preceq B$ denotes the Loewner order: $A \preceq B$ when $B - A$ is PSD (positive semi-definite), and similarly for $\succeq$.

We now define some basic notions used throughout this work.

\begin{definition}[Support of an operator]
    For an operator $P \in \mathbb{C}^{d^\qubits \times d^\qubits}$ on a system of $\qubits$ qudits, its \emph{support}, $\supp(P) \subseteq [\qubits]$ is the subset of qudits that $P$ acts non-trivially on.
    That is, $\supp(P)$ is the minimal set of qudits such that $P$ can be written as $P = O_{\supp(P)} \otimes \id_{[n] \setminus \supp(P)}$ for some operator $O$.
\end{definition}

\begin{definition}[Separable operator]
    For an operator $A \in \mathbb{C}^{d^\qubits \times d^\qubits}$ on $\qubits$ qudits, we say that $A$ is a \emph{product operator} if it can be expressed as a tensor product of operators over subsystems, $A = A^{(1)} \otimes \dots \otimes A^{(\qubits)}$.
    $A$ is \emph{separable} if it can be written as a sum of product operators $A = \sum_i (A_i^{(1)} \otimes \dots \otimes A_i^{(\qubits)})$ where every $A_i^{(j)}$ is positive semi-definite.
    We say a state is separable if its density matrix $\rho$ is separable.
\end{definition}

A property that we will use throughout is that the matrix exponential does not increase support size, which follows from its series expansion.

\begin{fact}\label{fact:trivial}
    For a square matrix $B$, $e^{\id \otimes B} = \id \otimes e^{B}$.  
\end{fact}

We will sometimes work in the qubit setting, $d = 2$, in which case it is useful to work with the space in the basis of (tensor products of) Pauli matrices.

\begin{definition}[Pauli matrices] \label{def:paulis}
    The Pauli matrices are the following $2 \times 2$ Hermitian matrices.
    \begin{equation*}
    \sigma_\id = \begin{pmatrix}
        1 & 0 \\ 0 & 1
    \end{pmatrix}, \qquad \sigma_x = \begin{pmatrix}
        0 & 1 \\
        1 & 0
    \end{pmatrix}, \qquad \sigma_y = \begin{pmatrix}
        0 & -\ii \\
        \ii & 0
    \end{pmatrix}, \qquad \sigma_z = \begin{pmatrix}
        1 & 0\\
        0& -1
    \end{pmatrix}.
    \end{equation*}
    These matrices are unitary and (consequently) involutory.
    Further, $\sigma_x \sigma_y = \ii \sigma_z$, $\sigma_y \sigma_z = \ii \sigma_x$, and $\sigma_z \sigma_x = \ii \sigma_y$, so the product of Pauli matrices is a Pauli matrix, possibly up to a factor of $\{\ii, -1, -\ii\}$.
    The non-identity Pauli matrices are traceless.
    We also consider tensor products of Pauli matrices, $P_1 \otimes \dots \otimes P_\qubits$ where $P_i \in \{\sigma_\id, \sigma_{x}, \sigma_{y}, \sigma_{z}\}$ for all $i \in [\qubits]$.
    The set of such products of Pauli matrices, which we denote $\locals$, form an orthogonal basis for the vector space of $2^\qubits \times 2^\qubits$ (complex) Hermitian matrices under the trace inner product.
    The product of two elements of $\locals$ is an element of $\locals$, possibly up to a factor of $\{\ii, -1, -\ii\}$.
\end{definition}

\subsection{Hamiltonians of interacting systems}

\begin{remark}[Hamiltonian input model] \label{rmk:ham-input}
    Recall the definition of a $\locality$-local Hamiltonian (\cref{def:hamiltonian}).
    When a Hamiltonian $H = \sum_{a = 1}^\terms H_a$ is given as input, we assume it is given as a list of the terms restricted to their support, along with their corresponding supports, $(\supp(H_1), (H_1)_{\supp(H_1)}), \dots, (\supp(H_\terms), (H_\terms)_{\supp(H_\terms)})$.
    This description has size linear in $\qubits$ and $\terms$, so it can be manipulated efficiently.

    We will primarily be concerned with algorithms in the qubit setting where the terms $H_a$ are rescaled Paulis $\lambda_a E_a$.
    Here, we can store a term $H_a$ as $\lambda_a$ and the Pauli $E_a = P_1 \otimes \dots \otimes P_\qubits \in \locals$ as a set $\{(i, P_i)\}_{i \in \supp(P)}$.
    When the parameters $\lambda_a$ are size $\bigO{1}$, this representation has size $\bigO{\locality \qubits}$.
    We will further assume that we have pre-processed it so that we can access it in various ways.
    In particular, given $i \in [\qubits]$, we will suppose we can get a list of terms $H_a$ such that $i \in \supp(H_a)$ in time linear in the size of the list.
    Further, given $a \in [\terms]$, we will suppose we can get a list of its neighbors in the dual interaction graph, the $b \in [\terms]$ such that $\supp(H_a) \cap \supp(H_b) \neq \varnothing$.
    These pre-processing steps can be done in $\bigOt{(\qubits + \terms)\locality \degree}$ time for a $\locality$-local Hamiltonian with dual interaction graph with degree $\degree$, and so will not dominate the running time.
\end{remark}

We will commonly need to consider Hamiltonians on various subsets of terms, which we notate in the following way.

\begin{definition}[Restricted Hamiltonian]
\label{def:restricted-ham}
For a Hamiltonian $H = \sum_{a=1}^\terms H_a$ on $n$ sites and a subset of sites $S\subseteq [n]$, we define the restricted Hamiltonian and its corresponding set of terms to be 
\[
    H^{(S)} = \sum_{a \in \mathcal{E}^{(S)}} H_a
    \quad \text{where} \quad \mathcal{E}^{(S)} = \braces{a \in [\terms] \mid \supp(H_a) \subseteq S}
\]
and the localized Hamiltonian to be
\[
    H_{(S)} = \sum_{a \in \mathcal{E}_{(S)}} H_a
    \quad \text{where} \quad \mathcal{E}_{(S)} = \braces{a \in [\terms] \mid \supp(H_a) \cap S \neq \varnothing}
    \,.
\]
We will also consider restricting a Hamiltonian to a set of terms, $H^{(\mathcal{Q})} = \sum_{a \in \mathcal{Q}} H_a$ when $\mathcal{Q} \subseteq [\terms]$.
\end{definition}

\subsection{Approximating the partition function}

At high temperature, we can efficiently estimate the partition function of a Gibbs state to multiplicative error.
This follows from cluster expansion; we defer most of the proof to the appendix.
Prior work immediately gives this result for $\beta \ll 1/\degree^2$ \cite{hkt21}; in the appendix we combine these results with tighter analyses \cite{fv17,mh21} to improve the critical temperature to $\bigTheta{1/\degree}$.

\begin{theorem}[Estimating the log-partition function]
\label{thm:estimation-log-partition}
    Let $H = \sum_{a=1}^\terms \lambda_a E_a \in \C^{2^\qubits \times 2^\qubits}$ be a $\locality$-local qubit Hamiltonian with Pauli terms, i.e.\ where $-1 \leq \lambda_a \leq 1$ and $E_a \in \locals$ is Pauli.
    Let $\degree$ be the degree of the dual interaction graph of $H$.
    Let $0 \leq \beta < \beta_c = 1/(100\degree)$.
    Then given any $0<\eta<1$, we can compute an estimate $\wh{z}$ such that 
    \begin{equation*}
        \log(\tr( e^{ -\beta H }  ) )  - \eta  \leq \wh{z} \leq \log(\tr(e^{ -\beta H  }) ) + \eta ,
    \end{equation*}
    in time 
    \begin{align*}
        \qubits \cdot (\qubits/\eta)^{ \frac{ 4+\log(\degree)} {\log(\beta_c/\beta) } }  \cdot \locality \cdot \degree^2 \cdot \polylog(\qubits/\eta). 
    \end{align*}
\end{theorem}

\begin{proof}[Proof of \cref{thm:estimation-log-partition}]
This follows from \cref{lem:cluster}.
Let $\beta_* = 1/(e(e+1)(1 + e(\degree - 1)))$ be the inverse temperature threshold from that lemma; since $\beta < \beta_c < \beta_*$, the results of that lemma hold for our choice of inverse temperature $\beta$.
In particular, we can write $\logpart = \log(\tr(e^{-\beta H}))$ as a power series
\begin{align*}
    \logpart = \sum_{\ell \geq 0} \beta^\ell p_\ell(\lambda),
\end{align*}
and when we estimate it by truncating at some level $k$,
\begin{align*}
    \widetilde{\logpart} = \sum_{\ell = 0}^k \beta^\ell p_\ell(\lambda),
\end{align*}
the error is bounded $\abs{\logpart - \widetilde{\logpart}} \leq \qubits \frac{(\beta/\beta_*)^{k+1}}{1 - \beta/\beta_*}$ by Part 2 of \cref{lem:cluster}.
This difference can be made to be at most $\eta$ by taking
\begin{align}
    k = \floor*{\frac{\log(n/((1 - \beta/\beta_*)\eta))}{\log(\beta_*/\beta)}}.
\end{align}
To compute $\widetilde{\logpart}$, we enumerate the list of monomials to order $k$ and then compute all corresponding coefficients.
The running time of this is dominated by the task at order $k$, which, using parts A and B of \cref{lem:cluster}, is bounded by
\begin{align}
    & \underbrace{\qubits(e\degree)^{k}}_{\text{\# of clusters}} \cdot \underbrace{(8^k + \locality)\poly(k)}_{\text{time to compute a coefficient}} + \underbrace{\bigO{k \degree \cdot \qubits(e\degree)^{k}}}_{\text{time to enumerate clusters}} \\
    &= \bigOt[\Big]{\qubits (8e\degree)^k \locality \degree} \\
    &= \bigOt[\Big]{\qubits\parens[\big]{\frac{n}{(1-\beta/\beta_*)\eta}}^{\frac{\log(8e\degree)}{\log(\beta_*/\beta)}}\locality \degree}
\end{align}
This gives the desired bound.
Since we are not optimizing the constant in $\beta_c$, for our choice of $\beta_c$, $\beta_c \leq \beta_* / 2$, so we can bound $(\frac{1}{1-\beta/\beta_*})^{\frac{\log(8e\degree)}{\log(\beta_* / \beta)}} = \bigO{\degree}$.
\end{proof}
\section{Separability of high-temperature Gibbs states}

In this section, we describe a method to write a Gibbs state $\rho$ as a positive linear combination of simple, ``configuration'' states.
For sufficiently small $\beta$, these configuration states are positive and separable, thereby proving the separability the corresponding Gibbs state.

We work with the following class of Hamiltonians, which is the straightforward generalization to qudits of the definitions given in the introduction, \cref{def:hamiltonian} and \cref{def:low-intersection-ham}.
\begin{definition}[Hamiltonian of a qudit system] \label{def:ham-qudit}
    Let $H = \sum_{a \in [\terms]} H_a$ be a Hamiltonian on $\qubits$ qudits with local dimension $d$, where the terms $H_a \in (\C^{d\times d})^{\otimes \qubits}$ are Hermitian matrices satisfying $\norm{H_a}_{\op} \leq 1$ for every $a \in [\terms]$.

    The dual interaction graph $\graph$ of $H$ has $[\terms]$ as the vertices and an edge between $a$ and $b$ if and only if $\supp(H_a)$ and $\supp(H_b)$ intersect.
\end{definition}

Though the terms $H_a$ refer to underlying matrices, for the results in this section we will typically be treating these as formal, non-commutating indeterminates.

\subsection{Low-degree polynomial approximation to a propagator}

We now consider the series expansion of $e^{-\beta H} e^{\beta(H - H^{(\mathcal{Q})})}$.
Some refer to this operator as a propagator~\cite{alhambra22}, because multiplying by this operator maps a Gibbs state on a restricted set of terms, $e^{-\beta(H - H^{(\mathcal{Q})})}$, to the full Gibbs state $e^{-\beta H}$.

\begin{theorem}[Propagator series] \label{thm:low-deg-approx}
    Let $H = \sum_{a \in [\terms]} H_a$ be a Hamiltonian with dual interaction graph $\graph$ with degree $\degree$.
    Then, for a subset of terms $\mathcal{Q} \subseteq [\terms]$, we can write
    \[
        e^{-\beta H} \cdot e^{\beta(H - H^{(\mathcal{Q})})} = \sum_{t = 0}^{\infty} \frac{\beta^t}{t!} f_{t}(H,H^{(\mathcal{Q})})
    \]
    where $f_t$ satisfies the recurrence $f_0(H,H^{(\mathcal{Q})}) = I$ and
    \begin{equation} \label{eq:restriction-recursion}
        f_{t+1}(H,H^{(\mathcal{Q})}) = -[H, f_t(H,H^{(\mathcal{Q})}) ] - f_t(H,H^{(\mathcal{Q})}) H^{(\mathcal{Q})} \,.
    \end{equation}
    Furthermore, let
    \begin{equation*}
        \mathcal{Q}^{(t)} = \Big\{ b \in [\terms]^t \,\Big|\,
        \text{every connected component of } \braces{b_1, \dots, b_t} \text{ in } \graph \text{ contains an element of } \mathcal{Q} \Big\}.
    \end{equation*}
    Then, for each $t > 0$, we can write
    \begin{align} \label{eq:dyson-bound}
        f_t(H,H^{(\mathcal{Q})}) = \sum_{b \in \mathcal{Q}^{(t)}} \mu_{(b_1,\dots,b_t)} H_{b_1}\dots H_{b_t}
        \text{ where } \sum_{b \in \mathcal{Q}^{(t)}} \abs{\mu_{(b_1,\dots,b_t)}} \leq \prod_{s=1}^t (\abs{\mathcal{Q}} + 2(\degree+1)s).
    \end{align}
\end{theorem}
\begin{proof}
We can expand out the series and write
\begin{align}
    e^{-\beta H} e^{\beta\parens{H - H^{(\mathcal{Q})}}}  
    & = \parens[\Big]{ \sum_{k = 0}^{\infty} \frac{ \beta^k (-H)^k  }{k! } }\parens[\Big]{ \sum_{\ell =0}^{\infty} 
    \frac{ \beta^{\ell} (H - H^{(\mathcal{Q})})^\ell  }{ \ell!  } }  \nonumber \\
    & = \sum_{t = 0}^{\infty } \frac{\beta^t }{t! } \sum_{k = 0}^t \frac{ (-H)^k \parens{H - H^{(\mathcal{Q})} }^{t -k } t! }{ k! (t-k)! } \\
    & = \sum_{t = 0}^{\infty } \frac{\beta^t }{t! } \underbrace{\sum_{k = 0}^t  \binom{t}{k} (-H)^k \parens{H-H^{(\mathcal{Q})}}^{t- k } }_{ f_{t}(H,H^{(\mathcal{Q})}) }. \nonumber
\end{align}
Now observe that, because $\binom{t}{k} = \binom{t-1}{k} + \binom{t-1}{k-1}$, $f_0(H,H^{(\mathcal{Q})}) = I$ and $f_t(H,H^{(\mathcal{Q})})$ satisfies the recurrence
\begin{align}
    f_t(H,H^{(\mathcal{Q})}) &= -H f_{t-1}(H,H^{(\mathcal{Q})}) + f_{t-1}(H,H^{(\mathcal{Q})}) (H - H^{(\mathcal{Q})}) \nonumber\\
    &= -[H,f_{t-1}(H,H^{(\mathcal{Q})})] - f_{t-1}(H,H^{(\mathcal{Q})}) H^{(\mathcal{Q})} \,.
\end{align}
This is the recurrence in the statement \eqref{eq:restriction-recursion}.
We will prove that $f_{t}$ satisfies \eqref{eq:dyson-bound} by induction.
First, we have that $f_1(H, H^{(\mathcal{Q})}) = H^{(\mathcal{Q})}$, which can be expressed as a sum of terms as specified in the theorem statement.
For the inductive step, assume that we have proven the desired properties up to some $t$.
Then we can write
\begin{align*}
    f_t(H,H^{(\mathcal{Q})}) = \sum_{b \in \mathcal{Q}^{(t)}} \mu_{(b_1,\dots,b_t)} H_{b_1}\dots H_{b_t} \,.
\end{align*}
The recurrence \eqref{eq:restriction-recursion} implies
\begin{align*}
    f_{t+1}(H,H^{(\mathcal{Q})}) = \sum_{b \in \mathcal{Q}^{(t)}} \mu_{(b_1,\dots,b_t)}\parens{-[H, H_{b_1}\dots H_{b_t}] - H_{b_1}\dots H_{b_t} H^{(\mathcal{Q})}} \,.
\end{align*}
By definition of $\mathcal{Q}^{(t)}$, $\braces{b_1, \dots, b_t}$'s connected components in $\graph$ each contain an element of $\mathcal{Q}$.
We now consider the above expression in two parts.
First, we can write 
\begin{align}
    [H, H_{b_1}\dots H_{b_t}]
    &= \sum_{a \in [\terms]} (H_a H_{b_1}\dots H_{b_t} - H_{b_1}\dots H_{b_t} H_a) \nonumber\\
    &= \sum_{\substack{a \in [\terms] \\ (a, b_1,\dots,b_t) \in \mathcal{Q}^{(t+1)}}} (H_a H_{b_1}\dots H_{b_t}) - \sum_{\substack{a \in [\terms] \\ (b_1,\dots,b_t, a) \in \mathcal{Q}^{(t+1)}}} (H_{b_1}\dots H_{b_t} H_a) \,.
    \label{eq:comm1}
\end{align}
The last step follows from noting that $[H_a, H_{b_1}\dots H_{b_t}]$ is zero unless the support of $H_a$ intersects that of $H_{b_1}\dots H_{b_t}$, which can only occur if $a$ is in the neighborhood of $\braces{b_1,\dots, b_t}$.
So, $(a, b_1, \dots, b_t) \in \mathcal{Q}^{(t+1)}$, since adding $a$ does not create a new connected component on $\graph$, and $(a, b_1, \dots, b_t) \in \mathcal{Q}^{(t+1)}$ if and only if $(b_1, \dots, b_t, a) \in \mathcal{Q}^{(t+1)}$.
Further, there are at most $(\degree + 1)t$ such $a$'s, since every $b_i$ has at most $\degree+1$ neighbors, so in \eqref{eq:comm1}, there are at most that many summands.
A similar argument works for the other part:
\begin{equation}\label{eq:comm2}
    H_{b_1}\dots H_{b_t} H^{(\mathcal{Q})} = \sum_{a \in \mathcal{Q}} H_{b_1}\dots H_{b_t} H_a,
\end{equation}
and we can conclude $H_{b_1}\dots H_{b_t} H_a$ is an element of $\mathcal{Q}^{(t+1)}$ because $a \in \mathcal{Q}$.
Further, there are $\abs{\mathcal{Q}}$ summands.
Combining the two parts \eqref{eq:comm1} and \eqref{eq:comm2}, we can write
\[
    f_{t+1}(H,H^{(\mathcal{Q})}) = \sum_{b' \in \mathcal{Q}^{(t+1)}} \mu_{(b_1',\dots,b_{t+1}')} H_{b_1'}\dots H_{b_{t+1}'},
\]
where, by the above argument and the inductive hypothesis,
\[
    \sum_{b' \in \mathcal{Q}^{(t+1)}} \abs{\mu_{(b_1',\dots,b_{t+1}')}} \leq \sum_{b \in \mathcal{Q}^{(t)}} \abs{\mu_{(b_1,\dots,b_{t})}} (\abs{\mathcal{Q}} + 2(\degree+1)t) \leq \prod_{s=1}^t (\abs{\mathcal{Q}} + 2(\degree+1)s),
\]
which completes the proof.
\end{proof}

In light of \cref{thm:low-deg-approx}, we define the truncation of this series:
\begin{definition}[Truncated propagator series]\label{def:truncated-series}
    For any integer $k \geq 0$ and parameter $\beta$, we define $T_{k,\beta}(H, H^{(\mathcal{Q})}) = \sum_{t = 0}^k \frac{\beta^t}{t!} f_t(H,H^{(\mathcal{Q})})$ where $f_t$ is as constructed in \cref{thm:low-deg-approx}: $f_0(H, H^{(\mathcal{Q})}) = \id$ and $f_{t+1}(H,H^{(\mathcal{Q})}) = -[H, f_t(H,H^{(\mathcal{Q})}) ] - f_t(H,H^{(\mathcal{Q})}) H^{(\mathcal{Q})}$.
\end{definition}

We now give an algorithm for sampling a single monomial from one of the polynomials in \cref{thm:low-deg-approx}.
This subroutine will be useful in our sampling algorithm later on, but our current goal is to prove separability, for which any algorithm suffices, regardless of its efficiency.
For this setting, we can think about this algorithm as providing a constructive way to express these polynomials as (positive) linear combinations of simple monomials.

\begin{mdframed}
\begin{algorithm}[Sampling a monomial of $f_k(H, H^{(\mathcal{Q})})$]
    \label{algo:recursive-sample-term}\mbox{}
    \begin{description}
    \item[Input:] A Hamiltonian $H = \sum_{a \in [\terms]} H_a$ with dual interaction graph $\graph$ with degree $\degree$; subset of terms $\mathcal{Q} \subseteq [\terms]$.
    \item[Input:] Integer $k \geq 0$.
    \item[Output:] Scalar $c \in \R$ and terms $(a_1, \dots, a_k) \in [\terms]^k$ such that, for $E = H_{a_1} \dots H_{a_k}$, $\E[c E] = f_k(H, H^{(\mathcal{Q})})$.
    \item[Operation:]
    \mbox{}
    \end{description}
    \begin{algorithmic}[1]
        \State Initialize $c_0 = 1$, and $b^{(0)} = (,) \in [\terms]^0$, the empty list; \\
        \Comment{The list $b^{(t)}$ implicitly represents the monomial $E_t = H_{b_1^{(t)}}\dots H_{b_t^{(t)}}$.}
        \For{$t = \{0,1, \dots , k - 1 \}$}
            \State Flip a coin which is heads with probability $\frac{t}{t+1}$;
            \If{coin is heads}
                \State Let $\mathcal{R}_t = \braces{a \in [\terms] \mid a \text{ is in the neighborhood of } \braces{b_s^{(t)}}_{s \in [t]} \text{ on } \graph}$;
                \State Sample an element $a \in \mathcal{R}_t$ uniformly at random;
                \State Sample a bit $\xi \in \braces{0, 1}$ uniformly at random;
                \State Set
                \[
                    c_{t+1} = \frac{(t+1)2\abs{\mathcal{R}_t}}{t} (-1)^\xi c_t\; , \;
                    b^{(t+1)} = \begin{cases}
                        a \cup b^{(t)} & \text{ if } \xi = 1 \\
                        b^{(t)} \cup a & \text{ if } \xi = 0
                    \end{cases};
                \]
                \Comment{This is an unbiased estimator for $-\frac{t+1}{t}[H, c_t E_t]$ over the randomness of $a$ and $\xi$;
                The notation $a \cup v$ and $v \cup a$ denotes prepending and appending to the list $v$.}
            \EndIf
            \If{coin is tails}
                \State Sample an element $a \in \mathcal{Q}$ uniformly at random;
                \State Set
                \begin{equation*}
                    c_{t+1} = -(t+1)\abs{\mathcal{Q}} c_t \; , \; b^{(t+1)} = b^{(t)} \cup a;
                \end{equation*}
                \Comment{This is an unbiased estimator for $-(t+1) c_t E_t H^{(\mathcal{Q})}$ over the randomness of $a$.}
            \EndIf
        \EndFor
        \State \Output $c_k, b^{(k)}$
    \end{algorithmic}
\end{algorithm}
\end{mdframed}

\begin{lemma}[Unbiased estimator of $f_k$]\label{lem:sample-term}
    In the same setting as \cref{thm:low-deg-approx}, for any integer $k \geq 0$, if we run \cref{algo:recursive-sample-term}, then it outputs a $c \in \R$ with $\abs{c} \leq (k!) \cdot (\max(2(\degree+1), \abs{\mathcal{Q}}))^k$ and $b \in \mathcal{Q}^{(k)}$ such that, for $E = H_{b_1}\dots H_{b_k}$,
    \[
        \E[cE] =  f_k(H,H^{(\mathcal{Q})}) \,.
    \]
    \end{lemma}
\begin{proof}
We prove by induction that, for every $t$ over the course of the algorithm, $c_t$ and $b^{(t)}$ satisfy the desired properties of the output: for example, $\E[c_t E_t] = f_t(H, H^{(Q)})$, where $E_t = H_{b_1^{(t)}} \dots H_{b_t^{(t)}}$ is the monomial implicitly maintained by $b^{(t)}$.
For the base case of $t = 0$, the algorithm sets $c_0 = 1$ and $E_0 = \id$, giving $1 \cdot \id = f_0(H, H_{(S)})$.
Now we show how to go from $t$ to $t+1$.
By the inductive hypothesis, we can assume that after $t$ iterations in the algorithm $c_t$ and $E_t$ satisfy $\E[c_tE_t] = f_t(H, H^{(\mathcal{Q})})$, $\abs{c_t} \leq (t!) \cdot (\max(2(\degree+1), \abs{\mathcal{Q}}))^t$, and $E_t \in \mathcal{Q}^{(t)}$.
Recall from \cref{thm:low-deg-approx} that
\[
    f_{t+1}(H,H^{(\mathcal{Q})}) = -[H, f_t(H,H^{(\mathcal{Q})}) ] - f_t(H,H^{(\mathcal{Q})}) H^{(\mathcal{Q})} \,.
\]
In iteration $t+1$ of \cref{algo:recursive-sample-term}, a coin is flipped, splitting the algorithm into two cases, the output when the coin is heads and the output when the coin is tails.

When the coin lands heads, taking expectation over the randomness in iteration $t+1$,
\begin{align}
    \E_{a, \xi}[c_{t+1} E_{t+1} \mid \text{heads}]
    &= \frac{1}{\abs{\mathcal{R}_t}}\sum_{a \in \mathcal{R}_t} \frac12 \parens[\Big]{\frac{(t+1)2\abs{\mathcal{R}_t}}{t}(-c_t H_a E_t + c_t E_t H_a)} \nonumber \\
    &= \frac{t+1}{t} \sum_{a \in \mathcal{R}_t} -[H_a, c_t E_t] \label{eq:expec1}\\
    &= \frac{t+1}{t} (-[H, c_t E_t]), \nonumber
\end{align}
where the last step follows from the same argument as for \eqref{eq:comm1}: the only terms $a \in [\terms]$ for which $[H_a, E_t]$ is non-zero are the ones which are distance at most one from $b^{(t)}$ on the graph $\graph$.
These are the terms in $\mathcal{R}_t$.
By a similar argument, $b^{(t+1)} \in \mathcal{Q}^{(t+1)}$.

When the coin lands tails, taking expectation over the randomness in iteration $t+1$,
\begin{align}
    \E_{a}[c_{t+1} E_{t+1} \mid \text{tails}]
    &= \frac{1}{\abs{\mathcal{Q}}} \sum_{a \in \mathcal{Q}} -(t+1)\abs{\mathcal{Q}} c_t E_t H_a \nonumber \\
    &= -(t+1) c_t E_t \sum_{a \in \mathcal{Q}} H_a \label{eq:expec2}\\
    &= -(t+1) c_t E_t H^{(\mathcal{Q})} \,. \nonumber
\end{align}
In this case, $b^{(t+1)}$ is $b^{(t)}$ with an element of $\mathcal{Q}$ added, and so is in $\mathcal{Q}^{(t+1)}$.
The coin lands heads with probability $\frac{t}{t+1}$ and tails with probability $\frac{1}{t+1}$, so putting \eqref{eq:expec1} and \eqref{eq:expec2} together, we have
\[
    \E_{\text{iteration }t+1}[c_{t+1}E_{t+1}] = -[H, c_t E_t ] - c_t E_t H^{(\mathcal{Q})} \,.
\]
By the inductive hypothesis, $c_tE_t$ was drawn from a distribution such that $\E[c_tE_t] = f_t(H,H_{(\mathcal{Q})})$, so $\E[c_{t+1}E_{t+1}] = f_{t+1}(H,H_{(\mathcal{Q})})$ as desired.

To get the bound on $c_{t+1}$, observe that in the heads case, $\abs{c_{t+1}} = (t+1)2 \frac{\abs{\mathcal{R}_t}}{t} \abs{c_t} \leq (t+1)2(\degree + 1)\abs{c_t}$, where the bound on $\abs{\mathcal{R}_t}$ arises because the number of terms $a$ neighboring some $b_s^{(t)}$ is bounded by $t(\degree + 1)$.
In the tails case, $\abs{c_{t+1}} = (t+1)\abs{\mathcal{Q}} \abs{c_t}$.
So, in either case
\begin{align*}
    \abs{c_{t+1}} \leq (t+1)\max(2(\degree + 1), \abs{\mathcal{Q}}) \abs{c_t} \leq (t+1)! \max(2(\degree + 1), \abs{\mathcal{Q}})^{t+1},
\end{align*}
where the last inequality uses the inductive hypothesis.
This completes the proof.
\end{proof}

Now that we have a sampler for $f_k$, we can build a sampler for $e^{-\beta H} e^{\beta(H - H^{(\mathcal{Q})})}$ (and its truncated version $T_{k, \beta}(H, H^{(\mathcal{Q})})$).

\begin{mdframed}
\begin{algorithm}[Sampling a monomial of $e^{-\beta H} e^{\beta(H - H^{(\mathcal{Q})})}$]
    \label{algo:sample-term}\mbox{}
    \begin{description}
    \item[Input:] A Hamiltonian $H = \sum_a H_a$ with dual interaction graph $\graph$ with degree $\degree$; temperature parameter $\beta \geq 0$; subset of terms $\mathcal{Q} \subseteq [\terms]$; threshold $t_{\max} \in \mathbb{N} \cup \braces{\infty}$.
        \item[Output:] Scalar $c \in \mathbb{R}$ and terms $b = (b_1, \dots, b_t)$ such that their product $E = H_{b_1}\dots H_{b_t}$ satisfies $\E[I + cE] = T_{t_{\max},\beta}(H, H_{(S)})$.
    \item[Operation:]
    \mbox{}
    \end{description}
    \begin{algorithmic}[1]
        \State Sample $t \sim \{0, 1, 2, \dots, t_{\max}\}$ with probabilities $\{2^{-t_{\max}}, 2^{-1}, 2^{-2}, \dots, 2^{-t_{\max}} \}$;
        \State If $t = 0$, set $c' = 0$ and $b = (,)$ the empty list;
        \State Otherwise, run \cref{algo:recursive-sample-term} on $H, \mathcal{Q}$ with parameter $k \leftarrow t$ to obtain $c'$ and $b = (b_1,\dots, b_t)$;
        \State \Output $c = \frac{2^t \beta^t}{t!} c'$ and $b = (b_1, \dots, b_t)$;
    \end{algorithmic}
\end{algorithm}
\end{mdframed}

\begin{lemma}[Unbiased estimator of $T_{t_{\max}, \beta}$]\label{claim:sample-term-unbiased}
The output of \cref{algo:sample-term}, $c$ and $(b_1, \dots, b_t)$, satisfies for $E = H_{b_1} \dots H_{b_t}$,
\[
    \E[I + c E] = T_{t_{\max}, \beta}(H, H^{(\mathcal{Q})}) \,.
\]
Further, if $t = 0$, then $c = 0$, and otherwise, $\abs{c} \leq (2\beta\max(2(\degree + 1), \abs{\mathcal{Q}}))^t$.
\end{lemma}
\begin{proof}
By linearity and the guarantees of \cref{lem:sample-term},
\[
    \E[I + c E] = I + \sum_{t = 1}^{t_{\max}} \frac{\beta^t}{t!} f_t(H,H^{(\mathcal{Q})}) =  T_{t_{\max}, \beta}(H, H^{(\mathcal{Q})})
\]
where the last step follows from \cref{def:truncated-series}.
Again using \cref{lem:sample-term},
\[
    \abs{c} = \frac{2^t\beta^t}{t!} \abs{c'}
    \leq \frac{2^t\beta^t}{t!} (t!)(\max(2(\degree + 1), \abs{\mathcal{Q}}))^t
    \leq (2\beta \cdot \max(2(\degree + 1), \abs{\mathcal{Q}}))^t,
\]
giving the coefficient bound.
\end{proof}

None of this argument so far has used a bound on $\beta$.
However, we want our estimator $\id + c E$ to always be a positive semi-definite matrix, which we can only guarantee if $\abs{c} \leq 1$ always.
So, this algorithm will only be useful when $\beta < \frac{1}{2\max(2(\degree + 1), \abs{\mathcal{Q}})}$.

\subsection{Additional structural properties of restricted Gibbs states}

In this section, we prove additional lemmas which are necessary for the state preparation algorithm in \cref{sec:preparable}.
This can safely be skipped to get to the separability proof.

In order to convert the algorithm in \cref{algo:separability} to one that samples from the Gibbs state $e^{-\beta H} / \tr e^{-\beta H}$, we need to account for the normalization by the partition function $\tr e^{-\beta H}$, as well as the error from not computing the entire series (which arises in the choice of $\eps$ in \cref{algo:single-step-2}).
We now prove some structural lemmas about Gibbs states which address these issues.

We first prove that the Gibbs state corresponding to $H - H_{(S)}$ has roughly the same spectrum as the Gibbs state corresponding to $H$, where $S$ is a small subset of sites.

\begin{lemma}[Spectrum of restricted Hamiltonians]\label{claim:psd-ordering}
    Let $H = \sum_{a \in [m]} H_a$ be a Hamiltonian with dual interaction graph $\graph$ with degree $\degree$.
    Let $S = \supp(H_{a^*})$ for some $a^* \in [\terms]$.
    If $\beta < \frac{1}{2C(\degree+1)}$ for some constant $C > 5$ then
    \[
        \parens[\Big]{1 - \frac{15}{C}} e^{-\beta H} \preceq e^{-\beta (H - H_{(S)})} \preceq \parens[\Big]{1 + \frac{15}{C}} e^{-\beta H}.
    \]
\end{lemma}
\begin{proof}
By \cref{thm:low-deg-approx}, we can write
\begin{align*}
    e^{\frac{\beta}{2} H} e^{-\frac{\beta}{2}(H - H_{(S)})}
    = \sum_{t = 0}^{\infty} \frac{(-\beta/2)^t}{t!} f_t(H, H_{(S)}),
\end{align*}
where $\norm{f_t(H, H_{(S)})}_{\op} \leq \prod_{s=1}^t ((\degree + 1)(2s + 1))$, using that the number of terms in $H_{(S)}$ is at most $\degree + 1$.
So, there is a matrix $A$ such that
\[
    e^{\frac{\beta}{2} H}e^{-\frac{\beta}{2} (H - H_{(S)})} = (I + A) \text{ with }
    \norm{A}_{\op} \leq \sum_{t = 1}^\infty (\beta(\degree + 1))^t \frac{\prod_{s=1}^t (2s+1)}{2^t(t!)} \leq \frac{5}{C} \,.
\]
Multiplying this by its conjugate, we have
\begin{align*}
    e^{\frac{\beta}{2} H} \cdot e^{-\beta(H - H_{(S)})} \cdot e^{\frac{\beta}{2} H}
    = (I + A)(I + A)^\dagger
\end{align*}
We are done upon concluding that, from the bound on $\norm{A}_{\op}$ and $C > 5$, $\norm{(I + A)(I + A^\dagger) - I}_{\op} \leq 2\norm{A}_{\op} + \norm{A}_{\op}^2 \leq 15/C$, so
\begin{equation*}
    \parens[\Big]{1 - \frac{15}{C}} I
    \preceq (I + A)(I + A)^\dagger
    \preceq \parens[\Big]{1 + \frac{15}{C}} I. \qedhere
\end{equation*}
\end{proof}

Next, we prove a sharper statement.
Recalling the construction of $T_{t, \beta/2}(H, H^{(\mathcal{Q})})$ in \cref{def:truncated-series}, we show that left and right multiplying a matrix by $e^{-\frac{\beta}{2} H}$ is close to left and right multiplying by $e^{-\frac{\beta}{2}(H - H^{(\mathcal{Q})})} T_{t, \beta/2}(H, H^{(\mathcal{Q})})^\dagger$ and its Hermitian conjugate. 

\begin{lemma}[Peeling the restricted Gibbs state]\label{claim:error-bound-2}
    Let $H = \sum_{a \in [m]} H_a$ be a $\locality$-local Hamiltonian with dual interaction graph with degree $\degree$.
    Let $S = \supp(a^*)$ for some $a^* \in [\terms]$.
    Let $P$ be a $2^n \times 2^n$ Hermitian matrix such that $0.5 \id \preceq P \preceq 2 \id$ and let $t \geq 0$ be an integer.
    Given any $\beta < \frac{1}{2C (\degree + 1) }$ for some constant $C > 5$, we have
    \begin{multline*}
        \parens[\Big]{1 -  \frac{100}{C^t}} e^{-\frac{\beta}{2} H} P e^{-\frac{\beta}{2} H} 
        \preceq  e^{-\frac{\beta}{2} (H - H_{(S)})} \cdot T_{t, \beta/2}(H, H_{(S)})^\dagger  \cdot P  \cdot T_{t, \beta/2}(H, H_{(S)}) \cdot e^{-\frac{\beta}{2} (H - H_{(S)})} \\
        \preceq \parens[\Big]{1 + \frac{100}{C^t}} e^{-\frac{\beta}{2} H} P e^{-\frac{\beta}{2} H} \,,
    \end{multline*}
    where $T_{t, \beta/2}$ is the truncation defined in \cref{def:truncated-series}.
\end{lemma}
\begin{proof}
It follows from \cref{thm:low-deg-approx} that 
\[
    e^{-\frac{\beta}{2} H} \cdot e^{\frac{\beta}{2}(H - H_{(S)})} = T_{t,\beta/2}(H, H_{(S)}) + E
\]
for some $E \in \C^{2^n \times 2^n}$ such that
\begin{align*}
    \norm{E}_{\op} \leq \sum_{s = t+1}^\infty (\beta (\degree + 1))^s \frac{\prod_{r=1}^s(2r+1)}{2^s(s!)} \leq \frac{5}{C^{t+1}}.
\end{align*}
We consider the expression
\begin{align}
    A &= P - e^{\frac{\beta}{2} H} \cdot e^{-\frac{\beta}{2}(H - H_{(S)})} \cdot  T_{t, \beta/2}(H, H_{(j)})^\dagger \cdot P \cdot  T_{t, \beta/2}(H, H_{(j)}) \cdot e^{-\frac{\beta}{2}(H - H_{(S)})} \cdot e^{\frac{\beta}{2} H} \nonumber \\
\intertext{
The lemma statement is equivalent to the statement that $-\frac{100}{C^t} P \preceq A \preceq \frac{100}{C^t} P$.
We continue:
}
    A &= P - (\id - e^{\frac{\beta}{2} H} \cdot e^{-\frac{\beta}{2}(H - H_{(S)})} \cdot E^\dagger ) \cdot P \cdot (\id - E \cdot e^{-\frac{\beta}{2}(H - H_{(S)})} \cdot e^{\frac{\beta}{2} H}) \nonumber \\
    &= e^{\frac{\beta}{2} H} \cdot e^{-\frac{\beta}{2}(H - H_{(S)})} \cdot E^\dagger \cdot P + P \cdot E \cdot e^{-\frac{\beta}{2}(H - H_{(S)})} \cdot e^{\frac{\beta}{2} H} \nonumber \\
    &\qquad - e^{\frac{\beta}{2} H} \cdot e^{-\frac{\beta}{2}(H - H_{(S)})} \cdot E^\dagger \cdot P \cdot E \cdot e^{-\frac{\beta}{2}(H - H_{(S)})} \cdot e^{\frac{\beta}{2} H}
    \label{eq:tmax-peel}
\end{align}
\cref{claim:psd-ordering} implies that $\norm{ e^{\frac{\beta}{2} H} e^{-\frac{\beta}{2}(H - H_{(S)})} }_{\op} \leq 4$.
Using this, along with the bounds on $\norm{E}_{\op}$ and $\norm{P}_{\op}$, we have that
\begin{align*}
    \norm{A}_{\op}
    &\leq 2\norm{ e^{\frac{\beta}{2} H} e^{-\frac{\beta}{2}(H - H_{(S)})} }_{\op} \norm{ E }_{\op} \norm{ P }_{\op} + \norm{ e^{\frac{\beta}{2} H} e^{-\frac{\beta}{2}(H - H_{(S)})} }_{\op}^2 \norm{ E }_{\op}^2 \norm{ P }_{\op} \\
    &\leq \frac{80}{C^{t+1}} + \frac{800}{C^{2t+2}}
    \leq \frac{50}{C^t}.
\end{align*}
Thus,
\[
    -\frac{50}{C^t} \id \preceq A \preceq \frac{50}{C^t} \id.
\]
We get the desired relations upon using that $\frac{1}{2} \id \preceq P \preceq 2 \id$.
\end{proof}

\subsection{Separability} \label{subsec:separability}

We now use our sampling procedure to prove separability at high temperature.
Our strategy is to show that $e^{-\beta H}$ can be written as a convex combination of simple matrices, which we call \emph{configurations}, all of which are positive semi-definite and separable.
Moreover, our proof is constructive, giving an algorithm for sampling configurations according to their weight in the combination.
We now define the objects our algorithm works with, which we call Hermitian monomials.

\begin{definition}[Hermitian monomial] \label{def:herm-monomial}
    For a Hamiltonian $H = \sum_{a=1}^\terms H_a$, we call an operator $X$ a \emph{Hermitian monomial} of (the terms of) $H$ if we can form it by starting from $X = \id$ and iteratively performing the following ``multiply and symmetrize'' operation: given $a_1,\dots, a_s$ and $b_1, \dots, b_t$, let
    \[
        X \gets \frac12\parens[\Big]{H_{a_1}\dots H_{a_s} X H_{b_t} \dots H_{b_1} + H_{b_1} \dots H_{b_t} X H_{a_s} \dots H_{a_1}}.
    \]
    We can maintain a description of a Hermitian monomial as a list of operations used to construct it.
    A Hermitian monomial $X$ can potentially be constructed in different ways, but we will abuse notation and use $X$ to refer to the matrix $X$ alongside an associated description.
    A description of a Hermitian monomial $X$ has an associated multiset of terms $\braces{a_1, \dots, a_t}$, which are the $H_a$'s which appear in the operations used to construct $X$.
    For example, the Hermitian monomial $\tfrac12(H_1 H_1 H_2 + H_2 H_1 H_1)$ has the associated multiset $\braces{1, 1, 2}$.
    We call $t$ the \emph{degree} of the monomial.\footnote{
        It is important to note that, for our algorithm, we will primarily be working with the \emph{description} of the monomial, and not the actual matrix $X$.
        Again, a matrix can have multiple different descriptions: for example, if $H_1^2 = \id$, then $\id$ can be described with a multiset of $\braces{1, 1}$, and with that description it would have a degree of $2$.
        This ambiguity will not cause issues, though; throughout, we will treat the $H_a$'s as essentially formal non-commutative indeterminates without any relations, and only at the end of the proof interpret these indeterminates as matrices.
    }
\end{definition}

A configuration is a tensor product of matrices $I + cX$, where $X$ is a Hermitian monomial.

\begin{definition}[Configuration] \label{def:configuration}
    For a Hamiltonian $H = \sum_{a = 1}^\terms H_a$, a \emph{configuration} is a collection $\mathcal{X} = \braces{(c_1, X_1), \dots, (c_\ell, X_\ell)}$ such that $c_i \in \R$, $X_{i}$ is a Hermitian monomial for every $i \in [\ell]$, and the $X_i$'s have disjoint supports.
    This configuration defines an associated operator $\sigma(\mathcal{X})$ which is $\id + c_i X_i$ on the support of $X_i$ and identity outside of the supports:
    \begin{align*}
        \sigma(\mathcal{X}) = (\id)_{[\qubits] \setminus (\supp(X_1) \cup \dots \cup \supp(X_\ell))} \otimes \parens[\Big]{\bigotimes_{i = 1}^\ell (\id + c_i X_i)_{\supp(X_i)}}.
    \end{align*}
\end{definition}

\begin{mdframed}
\begin{algorithm}[Pinning a single term]
    \label{algo:single-step-2}\mbox{}
    \begin{description}
    \item[Input:] Hamiltonian $H = \sum_a H_a$ with locality $\locality$ and dual interaction graph $\graph$ with degree $\degree$; inverse temperature parameter $\beta$ such that $\beta \leq \beta_c = 1/(100 \degree \locality)$; accuracy parameter $\eps \geq 0$.
    \item[Input:] Set of sites $S \subseteq [\qubits]$, configuration $\mathcal{X} = \braces{(c_1, X_1), \dots, (c_{\ell}, X_{\ell})}$ satisfying \eqref{eq:config-invariant}.
    \item[Output:] Set of sites $\wh{S} \subseteq [\qubits]$, configuration $\mathcal{\wh{X}} = \braces{(\wh{c}_1, \wh{X}_1), \dots, (\wh{c}_{\wh{\ell}}, \wh{X}_{\wh{\ell}})}$ satisfying \eqref{eq:config-invariant}.
    \item[Operation:]\mbox{}
    \begin{algorithmic}[1]
        \State Let $t_{\max} = 10 \log(\qubits / \eps)$, or $\infty$ if $\eps = 0$;
        \State Let $\wh{\ell} = \ell$:
        \State Let $\braces{a_1, \dots, a_t}$ be the terms associated to $(c_\ell, X_\ell)$, and choose an $a^* \in \mathcal{E}^{(S)}$ neighboring $\braces{a_1,\dots,a_t}$ in $\graph$;
        \label{line:extract-out}
        \Comment{For simplicity, these choices should be deterministic.}
        \If{no such $a^*$ exists}
            \State Add $(c_{\ell + 1}, X_{\ell + 1}) = (0, \id)$ to the configuration $\mathcal{X}$, and choose any $a^* \in \mathcal{E}^{(S)}$;
            \State Set $\wh{\ell} \gets \ell + 1$;
        \EndIf
                        \State Run \cref{algo:sample-term} twice (independently) with inputs $H \gets H^{(S)}$, $\beta \gets \beta/2$, $\mathcal{Q} \gets \mathcal{E}_{(a^*)} \cap \mathcal{E}^{(S)}$, and $t_{\max}$ to obtain $b_1, E_1, t_1$ and $b_2, E_2, t_2$; \label{line:sample-start}
        \State Let $\gamma = \frac{3}{5\locality}$;
        \State Sample $\xi \in \{0,1,2,3,4,5,6\}$ with probabilities $\braces{1 - \gamma, \frac{\gamma}{6}, \frac{\gamma}{6},\frac{\gamma}{6},\frac{\gamma}{6},\frac{\gamma}{6},\frac{\gamma}{6}}$. 
        \State Set $c_{\wh{\ell}} \gets \wh{c}$ and $X_{\wh{\ell}} \gets \wh{X}$, where
        \label{line:six-parts}
        \begin{itemize}
            \item If $\xi = 0$, then $\wh{c} = c_{\wh{\ell}}/(1 - \gamma)$ and $\wh{X} = X_{\wh{\ell}}$;
            \item If $\xi = 1$, then $\wh{c} = 6 b_1/\gamma$ and $\wh{X} = (E_1^\dagger  + E_1)/2$;
            \item If $\xi = 2$, then $\wh{c} = 6 b_2/\gamma$ and $\wh{X} = (E_2^\dagger + E_2)/2$;
            \item If $\xi = 3$, then $\wh{c} = 6 c_{\wh{\ell}}b_1 / \gamma$ and $\wh{X} = (E_1^\dagger X_{\wh{\ell}} + X_{\wh{\ell}} E_1)/2$;
            \item If $\xi = 4$, then $\wh{c} = 6 c_{\wh{\ell}}b_2 / \gamma$ and $\wh{X} = (E_2^\dagger X_{\wh{\ell}} + X_{\wh{\ell}} E_2)/2$;
            \item If $\xi = 5$, then $\wh{c} = 6 b_1b_2 / \gamma$ and $\wh{X} = (E_2^\dagger E_1 + E_1^\dagger E_2)/2$;
            \item If $\xi = 6$, then $\wh{c} = 6 c_{\wh{\ell}}b_1b_2 / \gamma$ and $\wh{X} = (E_2^\dagger X_{\wh{\ell}} E_1 + E_1^\dagger X_{\wh{\ell}} E_2)/2$;
        \end{itemize}
        \Comment{This is chosen so that $\E_{\xi}[\id + \wh{c} \wh{X}]$ takes a particular form \eqref{eq:why-6-cases}.}
        \State Set $S \gets S \setminus \supp(H_{a^*})$;
        \State \Output $S, \mathcal{X}$;
    \end{algorithmic}
    \end{description}
\end{algorithm}
\end{mdframed}

To analyze \cref{algo:single-step-2}, we will need the following arithmetic fact.
On the left-hand side are the operators which appear in our analysis, and on the right-hand side are Hermitian monomials.

\begin{fact} \label{fact:six-parts}
    For matrices $X$, $Y_1$, and $Y_2$, we can write
    \begin{multline*}
        (\id + Y_1)^\dagger (\id + X) (\id + Y_2)
        + (\id + Y_2)^\dagger (\id + X) (\id + Y_1) \\
        = 2\id
        + 2X
        + (Y_1^\dagger + Y_1)
        + (Y_2^\dagger + Y_2)
        + (Y_1^\dagger X + X Y_1)
        + (Y_2^\dagger X + X Y_2) \\
        + (Y_2^\dagger Y_1 + Y_1^\dagger Y_2)
        + (Y_2^\dagger X Y_1 + Y_1^\dagger X Y_2)\,.
    \end{multline*}
\end{fact}

\begin{mdframed}
\begin{algorithm}[Sampling a monomial of $e^{-\beta H}$]
    \label{algo:separability}\mbox{}
    \begin{description}
    \item[Input:] Hamiltonian $H = \sum_a H_a$ with dual interaction graph $\graph$ with degree $\degree$; inverse temperature parameter $\beta$ such that $\beta \leq \beta_c = 1/(100 \degree \locality)$.
    \item[Output:] Configuration $\mathcal{X}$ such that $\E[\sigma(\mathcal{X})] = e^{-\beta H}$.
    \item[Operation:]\mbox{}
    \begin{algorithmic}[1]
        \State Set $S = [\qubits]$;
        \State Set $\mathcal{X} = \varnothing$;
        \LComment{We maintain the invariant that $\E[e^{-\frac{\beta}{2}H^{(S)}}\sigma(\mathcal{X})e^{-\frac{\beta}{2}H^{(S)}}] = e^{-\beta H}$.}
        \While{$\mathcal{E}^{(S)} \neq \varnothing$}
            \State Run \cref{algo:single-step-2} with $S$ and $\mathcal{X}$ and $\eps = 0$; \label{line:while-loop}
            \State Take the output $\wh{S}$ and $\wh{\mathcal{X}}$ to be the new $S$ and $\mathcal{X}$;
        \EndWhile
        \State \Output $\mathcal{X}$;
    \end{algorithmic}
    \end{description}
\end{algorithm}
\end{mdframed}

The set $S$ corresponds to ``unpinned sites'', the set of sites which have not yet been sampled.
For the algorithm to work, $S$ must only interact with at most one of the monomials in the configuration, in the manner described in the following lemma.

\begin{lemma}[\cref{algo:single-step-2} is well-defined]
    \label{lem:well-defined}
    \cref{algo:single-step-2} is a well-defined procedure provided $\mathcal{E}^{(S)} \neq \varnothing$.
    Further, it maintains the following invariant: if the input $S$ and $\mathcal{X}$ satisfies that
    \begin{equation} \label{eq:config-invariant}
        \text{for all $i \in [\ell - 1]$, no term in $X_i$ is a neighbor of $\mathcal{E}^{(S)}$,}
    \end{equation}
    then the output $\wh{S}$ and $\mathcal{\wh{X}}$ satisfies the same condition.
    Finally, the output $\mathcal{\wh{X}}$ is a valid configuration.
\end{lemma}
\begin{proof}
The algorithm successfully completes provided there is an element of $\mathcal{E}^{(S)}$ so that a valid $a^*$ can be chosen.

For the invariant, we need to show that, for all $i \in [\wh{\ell} - 1]$, no term in $\wh{X}_i$ is a neighbor of $\mathcal{E}^{(\wh{S})}$.
First, observe that $\mathcal{\wh{X}}$ only differs from $\mathcal{X}$ in the last-indexed monomial in the configuration, which is either $i = \ell$ or $i = \ell + 1$.
Further, $\wh{S} \subset S$, so $\mathcal{E}^{(\wh{S})} \subseteq \mathcal{E}^{(S)}$, and so the invariant is maintained for $i < \ell$.
This proves the statement when $\wh{\ell} = \ell$, and the remaining case is when $\wh{\ell} = \ell + 1$, where we need to show that no term in $\wh{X}_\ell$ is a neighbor of $\mathcal{E}^{(\wh{S})}$.
This case only occurs when there are no such neighbors, as seen in \cref{line:extract-out}, proving the invariant.

The output $\mathcal{X}$ is a valid configuration because the input is a configuration, and $\mathcal{\wh{X}}$ only modifies $\mathcal{X}$ by adding terms from $\mathcal{E}^{(S)}$ to monomials $\ell$ or $\ell + 1$.
From the invariant, these terms have disjoint support from all other terms.
Since $E_1$ and $E_2$ are products of $H_a$'s, $\wh{X}$ is a Hermitian monomial when $X$ is, by definition (\cref{def:herm-monomial}).
\end{proof}

\begin{lemma} \label{lem:sep-expectation}
    \cref{algo:single-step-2} outputs a $\wh{S} \subseteq [\qubits]$ and configuration $\mathcal{\wh{X}}$ such that
    \begin{align*}
        \E[\sigma(\mathcal{\wh{X}})] = T_{t_{\max}, \beta/2}(H^{(S)}, H_{(S \setminus \wh{S})}^{(S)})^\dagger \sigma(\mathcal{X}) T_{t_{\max}, \beta/2}(H^{(S)}, H_{(S \setminus \wh{S})}^{(S)}) \, .
    \end{align*}
    In particular, when $\eps = 0$,
    \begin{align*}
        e^{-\frac{\beta}{2}H^{(\wh{S})}}\E[\sigma(\mathcal{\wh{X}})]e^{-\frac{\beta}{2}H^{(\wh{S})}} = e^{-\frac{\beta}{2} H^{(S)}} \sigma(\mathcal{X}) e^{-\frac{\beta}{2} H^{(S)}}\,.
    \end{align*}
\end{lemma}
\begin{proof}
Using linearity of expectation and \cref{fact:six-parts}, we can conclude that, with $\wh{c}$ and $\wh{X}$ as defined in \cref{line:six-parts} of the algorithm,
\begin{align}
    \E_\xi[\id + \wh{c}\wh{X}]
    &= \id + c_{\wh{\ell}}X_{\wh{\ell}}
    + b_1 (E_1^\dagger + E_1)/2
    + b_2 (E_2^\dagger + E_2)/2
    + c_{\wh{\ell}}b_1 (E_1^\dagger X_{\wh{\ell}} + X_{\wh{\ell}} E_1)/2 \nonumber\\
    & \hspace{2em} + c_{\wh{\ell}}b_2 (E_2^\dagger X_{\wh{\ell}} + X_{\wh{\ell}} E_2)/2
    + b_1b_2 (E_2^\dagger E_1 + E_1^\dagger E_2)/2
    + c_{\wh{\ell}}b_1b_2 (E_2^\dagger X_{\wh{\ell}} E_1 + E_1^\dagger X_{\wh{\ell}} E_2)/2 \nonumber\\
    &= \frac12\parens[\Big]{(\id + b_1E_1)^\dagger(\id + c_{\wh{\ell}}X_{\wh{\ell}})(\id + b_2 E_2)
    + (\id + b_2 E_2)^\dagger(\id + c_{\wh{\ell}}X_{\wh{\ell}})(\id + b_1 E_1)}\,. \label{eq:why-6-cases}
\end{align}
\cref{claim:sample-term-unbiased} implies that the outputs of \cref{algo:sample-term} in \cref{line:sample-start} satisfy
\begin{align*}
    \E[\id + b_1 E_1] = \E[\id + b_2 E_2] = T_{t_{\max}, \beta/2}(H^{(S)}, H_{(a^*)}^{(S)})\,,
\end{align*}
where $H_{(a^*)}^{(S)}$ is the sum of all terms $H_a$ where $\supp(H_a) \subseteq S$ and $\supp(H_a) \cap \supp(H_{a^*})$ is non-empty.
So, taking the expectation over the entire algorithm and using that $\id + b_1E_1$ and $\id + b_2E_2$ are independent,
\begin{align*}
    \E[\id + \wh{c} \wh{X}] &= \E\bracks[\Big]{\frac12\parens[\Big]{(\id + b_1E_1)^\dagger(\id + c_{\wh{\ell}}X_{\wh{\ell}})(\id + b_2 E_2) + (\id + b_2 E_2)^\dagger(\id + c_{\wh{\ell}}X_{\wh{\ell}})(\id + b_1 E_1)}} \\
    &= T_{t_{\max}, \beta/2}(H^{(S)}, H_{(a^*)}^{(S)})^\dagger (\id + c_{\wh{\ell}}X_{\wh{\ell}}) T_{t_{\max}, \beta/2}(H^{(S)}, H_{(a^*)}^{(S)}).
\end{align*}
By \cref{lem:well-defined}, the support of $\id + \wh{c}\wh{X}$ is always disjoint from all of the other monomials in the configuration $\mathcal{\wh{X}}$, so
\begin{align*}
    \E[\sigma(\mathcal{\wh{X}})] &= T_{t_{\max}, \beta/2}(H^{(S)}, H_{(a^*)}^{(S)})^\dagger \sigma(\mathcal{X}) T_{t_{\max}, \beta/2}(H^{(S)}, H_{(a^*)}^{(S)}).
\end{align*}
When $\eps = 0$, this can be written as
\begin{align*}
    \E[\sigma(\mathcal{\wh{X}})]
    &= e^{\frac{\beta}{2}(H^{(S)} - H_{(a^*)}^{(S)})}e^{-\frac{\beta}{2} H^{(S)}} \sigma(\mathcal{X})  e^{-\frac{\beta}{2} H^{(S)}} e^{\frac{\beta}{2}(H^{(S)} - H_{(a^*)}^{(S)})} \\
    &= e^{\frac{\beta}{2}H^{(\wh{S})}}e^{-\frac{\beta}{2} H^{(S)}} \sigma(\mathcal{X}) e^{-\frac{\beta}{2} H^{(S)}} e^{\frac{\beta}{2}H^{(\wh{S})}},
\end{align*}
since $\wh{S} = S \setminus \supp(H_{a^*})$.
\end{proof}

What we have so far is enough to show that \cref{algo:separability} indeed outputs a configuration $\mathcal{X}$ with $\E[\sigma(\mathcal{X})] = e^{-\beta H}$.
However, for separability, we need more: namely, we need that $\sigma(\mathcal{X})$ is always positive semi-definite and separable.
We handle the former with the following lemma.
This is the key lemma where the critical temperature is used.
It is nontrivial to show, but we are able to prove it via a carefully chosen potential function.

\begin{lemma}[Coefficients remain small]\label{lem:coeff-potential}
    Let $H = \sum_a H_a$ be a $\locality$-local Hamiltonian and let the temperature parameter $\beta$ satisfy $\beta \leq \beta_c = 1 / (50 (\degree+1) \locality)$, where $\degree$ is the degree of the dual interaction graph of $H$.
    Consider running \cref{algo:single-step-2} on the input $S$ and $\mathcal{X}$, and receiving $\wh{S}$ and $\mathcal{\wh{X}}$ as output.
    Then $\abs{\wh{S}} < \abs{S}$.
    Further, consider the following invariant on the input:
    \begin{multline} \label{eq:potential}
                \text{for all } (c, X) \in \mathcal{X},\,
        \abs{c} \leq \parens{1-\gamma}^{\abs{S \cap (\supp(H_{a_1}) \cup \dots \cup \supp(H_{a_t}))}} \parens[\Big]{\frac{\beta}{\beta_c}}^t \\
        \text{where } \braces{a_1,\dots,a_t} \text{ are the terms associated to the Hermitian monomial } X
        \,.
    \end{multline}
    If the input satisfies \eqref{eq:potential}, then so does the output:
    \begin{align} \label{eq:potential-out}
                \abs{\wh{c}} \leq \parens{1-\gamma}^{\abs{\wh{S} \cap (\supp(H_{\wh{a}_1}) \cup \dots \cup \supp(H_{\wh{a}_{\wh{t}}}))}} \parens[\Big]{\frac{\beta}{\beta_c}}^{\wh{t}}
    \end{align}
\end{lemma}
\begin{proof}
    Because $\wh{S} = S \setminus \supp(H_{a^*})$ and the manner in which $a^*$ is chosen, $\abs{\wh{S}} < \abs{S}$.

    To prove the inductive invariant, it suffices to consider the term $(\wh{c}_{\wh{\ell}}, \wh{X}_{\wh{\ell}})$, which we denote $(\wh{c}, \wh{X})$ in the algorithm, since this is the only term that \cref{algo:single-step-2} changes: the invariant will continue to hold provided that the term does not change.
    We first note that, in \cref{line:sample-start}, \cref{algo:sample-term} is always run for $\beta \gets \beta/2$ and $\mathcal{Q} \gets \mathcal{E}_{(a^*)} \cap \mathcal{E}^{(S)}$ (see \cref{def:restricted-ham}).
    By definition of the dual interaction graph, $\abs{\mathcal{Q}} \leq \degree + 1$.
    So, by \cref{claim:sample-term-unbiased},
    \begin{align*}
        &b_1 = 0 \text{ and } t_1 = 0 \text{ or } \abs{b_1} \leq (\beta \cdot 2(\degree + 1))^{t_1}, \text{ and} \\
        &b_2 = 0 \text{ and } t_2 = 0 \text{ or } \abs{b_2} \leq (\beta \cdot 2(\degree + 1))^{t_2}
    \end{align*}
    We now show that $\wh{c}$ is small for every choice of $\xi$.
    For this proof, we denote $c = c_{\wh{\ell}}$ and $X = X_{\wh{\ell}}$ to be the term which is to be changed.
    Further, we will abuse notation and, for a Hermitian monomial $X$ with associated terms $\braces{a_1, \dots, a_t}$, write $\supp(X) = \supp(H_{a_1}) \cup \dots \cup \supp(H_{a_t})$.
    This set contains the ``true'' support of $X$, but can be larger in general.
    So, we need to show that, assuming the invariant \eqref{eq:potential} which states that $\abs{c} \leq (1-\gamma)^{\abs{S \cap \supp(X)}}(\frac{\beta}{\beta_c})^{t}$, then $\abs{\wh{c}} \leq (1-\gamma)^{\abs{\wh{S} \cap \supp(\wh{X})}}(\frac{\beta}{\beta_c})^{\wh{t}}$.
    Note that, by our choice of $\beta_c$, $10 \locality \beta_c \cdot 2(\degree + 1) \leq \frac25 = 1 - \locality\gamma$.
    \paragraph{Case $\xi = 0$:}
    Here, $\wh{X} = X$, the associated terms are identical ($\wh{a}_s = a_s$ for all $s \in [t]$) and $\wh{c} = c / (1 - \gamma)$.
    Then either $c = 0$, and the bound holds trivially, or $a^*$ was chosen to neighbor some $a_s$ associated to $X$.
    In this case, the scalar increases; however, this is accounted for in the potential, because $\supp(H_{a^*})$ is removed from $S$ to get $\wh{S}$:
    \begin{align*}
        \abs{\wh{c}} = \abs{c} / (1 - \gamma)
        \leq (1 - \gamma)^{\abs{S \cap \supp(X)} - 1}\parens[\Big]{\frac{\beta}{\beta_c}}^t
        \leq (1 - \gamma)^{\abs{\wh{S} \cap \supp(X)}}\parens[\Big]{\frac{\beta}{\beta_c}}^t \,.
    \end{align*}
    The first inequality uses the induction hypothesis; the second uses that $\wh{S} \subseteq S$.
    \paragraph{Case $\xi = 1$, $\xi = 2$, $\xi = 5$:}
    We first work out the $\xi = 1$ case.
    Here, $\wh{X} = (E_1^\dagger + E_1) / 2$, and so the monomial has degree $\wh{t} = t_1$.
    If $t_1 = 0$, then $b_1 = 0$ and we are done.
    Otherwise,
    \begin{align*}
        \abs{\wh{c}} &= 6\abs{b_1}/\gamma
        \leq 10 \locality (\beta \cdot 2(\degree + 1))^{t_1}
        \leq (10 \locality \beta_c \cdot 2(\degree + 1))^{t_1} \parens[\Big]{\frac{\beta}{\beta_c}}^{t_1} \\
        &\leq (1 - \locality \gamma)^{t_1} \parens[\Big]{\frac{\beta}{\beta_c}}^{t_1}
        \leq (1 - \gamma)^{\locality t_1} \parens[\Big]{\frac{\beta}{\beta_c}}^{t_1}
        \leq (1 - \gamma)^{\abs{\wh{S} \cap \supp(\wh{X})}} \parens[\Big]{\frac{\beta}{\beta_c}}^{\wh{t}}\,.
    \end{align*}
    The $\xi = 2$ case is analogous, switching out $E_1$ and $t_1$ for $E_2$ and $t_2$.
    The $\xi = 5$ case is also analogous, with $E_1^\dagger E_2$ and $\wh{t} = t_1 + t_2$.
    \paragraph{Case $\xi = 3$, $\xi = 4$, $\xi = 6$:}
    We first work out the $\xi = 3$ case.
    Here, $\wh{X} = (E_1^\dagger X + X E_1) / 2$, so the monomial has degree $t + t_1$, where $t$ is the degree of the original monomial $X$, and the associated terms are the terms from $X$, along with $t_1$ additional terms from $E_1$.
    If $t_1 = 0$ then $\wh{c} = 6cb_1 / \gamma = 0$ so we are done.
    Otherwise,
    \begin{align*}
        \abs{\wh{c}} &= 6\abs{c}\abs{b_1}/\gamma \\
        &\leq 10\locality (1-\gamma)^{\abs{S \cap \supp(X)}}\parens[\Big]{\frac{\beta}{\beta_c}}^{t} (\beta \cdot 2(\degree + 1))^{t_1} \\
        &\leq (1-\gamma)^{\abs{S \cap \supp(X)} + \locality t_1} \parens[\Big]{\frac{\beta}{\beta_c}}^{t + t_1} \\
        &\leq (1-\gamma)^{\abs{\wh{S} \cap \supp(\wh{X})}} \parens[\Big]{\frac{\beta}{\beta_c}}^{\wh{t}}\,. 
    \end{align*}
    The $\xi = 4$ and $\xi = 6$ cases are analogous, in the same manner as the $\xi = 2$ and $\xi = 5$ cases.
\end{proof}

\begin{theorem} \label{thm:sep-algorithm}
    Let $H = \sum_a H_a$ be a Hamiltonian with locality $\locality$ and dual interaction graph $\graph$ with degree $\degree$, and let the temperature parameter $\beta$ satisfy $\beta \leq \beta_c = 1/(100 \locality \degree)$.
    Then \cref{algo:separability} outputs a configuration $\mathcal{X}$ such that $\E[\sigma(\mathcal{X})] = e^{-\beta H}$ and, for all $(c, X) \in \mathcal{X}$, $\abs{c} \leq (\beta / \beta_c)^t$, where $t$ is the degree of the Hermitian monomial $X$.
\end{theorem}
\begin{proof}
We show that through the while loop, \cref{algo:separability} maintains the two invariants that
\begin{align*}
    \E[e^{-\frac{\beta}{2}H^{(S)}}\sigma(\mathcal{X}) e^{-\frac{\beta}{2}H^{(S)}}] = e^{-\beta H}
\end{align*}
and, from \cref{eq:potential}, for all $(c, X) \in \mathcal{X}$ with degree $t$,
\begin{align*}
    \abs{c} \leq \parens{1-\gamma}^{\abs{S \cap \supp(X)}} \parens[\Big]{\frac{\beta}{\beta_c}}^t.
\end{align*}
(We use the same notation as in \cref{lem:coeff-potential}, where $\supp(X)$ denotes the support of the terms associated to $X$.)
Since the initialization is that $S = [\qubits]$ and $\mathcal{X} = \varnothing$, so that $\sigma(\mathcal{X}) = \id$, both invariants hold before the while loop begins.
\cref{lem:sep-expectation} implies that the first invariant is maintained, since the output of \cref{algo:single-step-2} satisfies
\begin{align*}
    e^{-\frac{\beta}{2}H^{(\wh{S})}}\E_{\text{iteration}}[\sigma(\wh{\mathcal{X}})]e^{-\frac{\beta}{2}H^{(\wh{S})}} = e^{-\frac{\beta}{2} H^{(S)}} \sigma(\mathcal{X}) e^{-\frac{\beta}{2} H^{(S)}}\,,
\end{align*}
where the randomness here is taken only over the call to the algorithm in \cref{line:while-loop}.
\cref{lem:coeff-potential} implies that the second invariant is maintained.
Finally, by \cref{lem:coeff-potential}, the size of $S$ decreases every iteration, so the algorithm terminates after at most $\qubits$ iterations.
Upon termination, $S = \varnothing$, so the invariants imply that, for the output configuration $\mathcal{X}$, $\E[\sigma(\mathcal{X})] = e^{-\beta H}$ and, for every $(c, X) \in \mathcal{X}$ with degree $t$, $\abs{c} \leq (\frac{\beta}{\beta_c})^t$.
\end{proof}

Because Hermitian monomials $X$ satisfy $\norm{X}_{\op} \leq 1$ and the scalars in the output configuration $\mathcal{X}$ always satisfy $\abs{c} \leq 1$, $\id + c X$, and thus $\sigma(\mathcal{X})$ is always positive semi-definite.
Further, for reasonable classes of terms, $\sigma(\mathcal{X})$ can be shown to be separable, which means that $e^{-\beta H}$ and therefore the Gibbs state $e^{-\beta H} / \tr e^{-\beta H}$ is separable.

\begin{lemma} \label{lem:pauli-monomial}
    Let $H = \sum_a H_a \in \C^{2^\qubits \times 2^\qubits}$ be a Hamiltonian over $\qubits$ qubits with Pauli terms, meaning that $H_a = \lambda_a E_a$ for $E_a \in \locals$ and $-1 \leq \lambda_a \leq 1$.
    Let $X$ be a Hermitian monomial of degree $t$ with respect to the terms of $H$.
    Then $\id + cX$ is separable when $-1 \leq c \leq 1$.
    Further, there is an algorithm which outputs a stabilizer product state $\ket{\psi}$ such that $\E[\ketbra{\psi}{\psi}] = (\id + cX) / \tr(\id + cX)$ in $\bigOt{\qubits + t \locality}$ time.
\end{lemma}
\begin{proof}
Since $X$ is a Hermitian monomial and every $H_a$ is a tensor product of Paulis, $X$ is either zero or a tensor product of Paulis, up to a sign.
This can be proved by induction on the definition of a Hermitian monomial (\cref{def:herm-monomial}).
With this, we can conclude that $\id + c X$ is separable: for $X = X_1 \otimes \dots \otimes X_\qubits$, we can write it as a convex combination of tensor products of PSD matrices by writing $\id + c X = (1 - \abs{c})\id + \abs{c}(\id + \frac{c}{\abs{c}}X)$, and then for the latter part, using
\begin{align*}
    \id + X = \frac{1}{2^\qubits}\bigotimes_{i=1}^\qubits \parens[\Big]{(\id + X_i) + (\id - X_i)} + \frac{1}{2^\qubits}\bigotimes_{i=1}^\qubits \parens[\Big]{(\id + X_i) - (\id - X_i)}.
\end{align*}
By expanding this, we conclude that, by drawing a uniformly random $s_1,\dots,s_\qubits \in \{\pm 1\}$ with even parity, i.e.\  $s_1s_2\dots s_\qubits = 1$, and preparing the state $\bigotimes_{i=1}^\qubits (\id + s_i X_i)$, this state is $\id + X$ in expectation.

To turn this into an algorithm, we note that given $X$ represented as a list of terms and operations, we can compute $X_1, \dots, X_\qubits$ efficiently.
For the qubits $i$ for which $X_i = \id$, we can take $\ket{\psi}_i$ to be $\ket{0}$ or $\ket{1}$ with half probability, thereby making the corresponding density matrix $\id / 2$.
On the support of $X$, we can draw a uniformly random string with even parity and prepare the corresponding state $\id + s_i X_i$; since $X_i$ is a non-identity Pauli, this is a pure eigenstate of a Pauli.
This gives an algorithm for preparing $(\id + cX) / \tr(\id + cX)$ by either preparing $\id + \frac{c}{\abs{c}} X$ or $\id$ with the appropriate probabilities.
\end{proof}

With this, we can prove separability of high-temperature Gibbs states for qubit systems, as stated in the introduction.
\begin{proof}[Proof of \cref{thm:separable}]
\cref{thm:sep-algorithm} implies that, for $\beta \leq 1/(100 \degree \locality)$, $e^{-\beta H}$ is a convex combination of matrices of the form $\sigma(\mathcal{X})$, where for every $(c, X) \in \mathcal{X}$, $\abs{c} \leq 1$ and $X$ is a Hermitian monomial (\cref{def:herm-monomial}).
By \cref{lem:pauli-monomial}, $(\id + cX)_{\supp(X)}$ is separable (and expressible as a positive linear combination of stabilizer product states), and since $\sigma(\mathcal{X})$ is the tensor product of such matrices, it is also separable (and expressible as a positive linear combination of stabilizer product states).
Since separability is closed under convex combinations and scaling by a constant factor, $e^{-\beta H} / \tr e^{-\beta H}$ is separable as desired.
Further, this argument shows it is expressible as a convex combination of stabilizer product states.
\end{proof}

More generally, this argument can prove separability for any set of terms $\braces{H_a}$ such that the corresponding class of Hermitian monomials $X$ satisfies that $\id + cX$ is separable.
There is some ``wiggle room'' in the argument, which is unused for the case of Pauli terms: \cref{thm:sep-algorithm} gives that $\abs{c} \leq (\beta / \beta_c)^{t}$ where $t$ is the degree of $X$, so with smaller choices of $\beta$, monomials with larger support are weighted smaller, which is useful for demonstrating separability.
For example, we can show that local Hamiltonians over qudits also have separable Gibbs states at high temperature.

\begin{theorem}[High-temperature Gibbs states are separable]
    \label{thm:gibbs-separable}
    Let $H = \sum_a H_a \in \C^{d^\qubits \times d^\qubits}$ be a $(\locality, \degree)$-low-intersection Hamiltonian over $\qubits$ qudits (\cref{def:ham-qudit}).
    Further suppose that every $H_a$ satisfies $\norm{H_a}_{\op} \leq 1$ and is a product operator, meaning that we can write $H_a = H_a^{(1)} \otimes \dots \otimes H_a^{(\qubits)}$.
    Then the Gibbs state $e^{-\beta H} / \tr e^{-\beta H}$ is separable for $\beta \leq 1/(100 \locality \cdot 4^\locality \cdot \degree)$.
\end{theorem}
\begin{proof}
\cref{thm:sep-algorithm} implies that $e^{-\beta H}$ is a convex combination of matrices of the form $\sigma(\mathcal{X})$, where for every $(c, X) \in \mathcal{X}$, $X$ is a Hermitian monomial and $\abs{c} \leq 4^{-\locality t}$, where $t$ is the degree of $X$.
It suffices to show that such $\id + cX$ are separable.
This holds because of the following identity relating the sum of a product operator and its conjugate transpose with the sum of two Hermitian product operators.
\begin{align*}
    A \otimes B + A^\dagger \otimes B^\dagger
    = (A + A^\dagger) \otimes (B + B^\dagger) + (\ii A + (\ii A)^\dagger) \otimes (-\ii B + (-\ii B)^\dagger).
\end{align*}
By iterating this, we can conclude that a Hermitian monomial with support size $k$ can be written as a sum of $2^{k}$ Hermitian product operators, each of which has operator norm bounded by $2^{k-1}$.
A Hermitian monomial $X$ with degree $t$ has support size at most $t \locality$, so when $\abs{c} \leq 2^{-(2t\locality-1)}$, $\id + cX$ is separable, as desired.
\end{proof}

\section{Fast state preparation} \label{sec:preparable}

In this section, we prove \cref{thm:main-gibbs}.
Here, we take \cref{algo:separability} and \cref{algo:single-step-2} from the separability proof, and then argue that we can run a Markov chain on the ``state tree'' of these algorithms to get a distribution over output configuration which satisfies $\E[\sigma(\mathcal{X}) / \tr(\sigma(\mathcal{X}))] \approx \rho$, instead of the original guarantee that $\E[\sigma(\mathcal{X})] = e^{-\beta H}$.

\subsection{Random walks on trees}
\label{sec:random-walks-on-trees}

As a subroutine of our main sampling algorithm, we will design a (classical) random walk on a tree.
In this section, we present some general machinery for analyzing the mixing times of random walks on trees.
We begin with the definition of a \textit{weighted tree}.

\begin{definition}[Weighted tree]
    Let $\tree$ be a tree with a unique root such that all root-to-leaf paths have length at most $n$.
    A weighted tree $(\tree,w)$ of depth $n$ is obtained by assigning some non-negative weight $w_v$ to each leaf $v$ of the tree.
    For an interior node $v$, we let the weight assigned to $v$ be the sum of the weights of the leaves in the sub-tree rooted at $v$.
\end{definition}

Next, we assume access to the following sampling sub-routine: 

\begin{definition}[Sample query]
\label{def:sampling-sub-tree}
For a weighted tree $(\tree, w)$, we define a sample query as querying a node $u$ and outputting a random child of $u$, where the probability of outputting a child $v$ is $w_v / w_u$.
\end{definition}

Since the weights on each leaf are non-negative, they induce a probability distribution over the leaves, corresponding to the distribution formed by repeatedly performing sample queries starting from the root node.

\begin{definition}[Leaf distribution]
\label{def:leaf-distribution}
For a weighted tree $(\tree, w)$, its leaf distribution is the distribution over leaves where each leaf $v$ is sampled proportional to its weight $w_v$.
\end{definition}

We then recall the definition of a Markov chain and the corresponding stationary distribution.

\begin{definition}[Transition matrix and stationary distribution]
\label{def:transition-matrix-stat-distribution}
A Markov chain on $N$ states is given by a transition matrix $P$ with entries $P_{ij}$ for $i,j \in [N]$ given by the probability of moving from state $i$ to state $j$.  We define the stationary distribution, denoted by $\pi = (\pi_1, \dots , \pi_N)$, such that $P \pi = \pi$.  
\end{definition}

\begin{definition}[Ergodic and time-reversible Markov chain]
\label{def:ergodic-markov-chain}
A Markov chain $P$ is \textit{ergodic} if there exists a positive integer $z$ such that $P^z$ is entry-wise positive.  It is \textit{time-reversible} if its stationary distribution $\pi$ satisfies
\[
P_{ij} \pi_i = P_{ji} \pi_j
\]
for all $i,j \in [N]$.
\end{definition}

Next, we define the notion of conductance of a Markov chain.

\begin{definition}[Conductance]
\label{def:conductance}
Given a Markov chain on $N$ states with transition matrix $P$ and stationary distribution $\pi$, for any subset $S \subseteq [N]$, the conductance $\Phi_S$ is defined by 
\[
\Phi_S = \frac{\sum_{i \in S, j \notin S} P_{ij} \pi_i}{\sum_{i \in S} \pi_i} \,.
\]
The global conductance of the chain is defined by $\Phi = \min_{S : C_S \leq 1/2} \Phi_S$, where $C_S = \sum_{ i \in S} \pi_i $.  
\end{definition}

A classical result of Jerrum and Sinclair~\cite{sj89} bounds the spectral gap of an ergodic, time-reversible Markov chain as a function of the conductance.

\begin{lemma}[Spectral gap of a Markov chain~\cite{sj89}]\label{lem:weighted-cheeger}
For an ergodic time-reversible Markov chain $P$, if we order the eigenvalues of $P$ as $\lambda_1 \geq \lambda_2 \geq \dots \geq \lambda_N$ where $\lambda_1 = 1$, then
\[
\lambda_1 \leq 1 - \frac{\Phi^2}{2} \,.
\]
\end{lemma}

The main result in this section, stated in \cref{lem:sample-from-tree} below, is about relating two different weighted trees on the same vertex set, which we denote by $(\tree,w)$ and $(\tree,w')$.  For a vertex $v$, the distortion between the two weight functions is just $w_v/w_v'$.  Note that scaling a weight function by a constant factor doesn't affect any of the resulting distributions.  Given any edge $(u,v)$ on this tree, we assume that the distortion between $w$ and $w'$ along this edge is bounded: $0.1 \leq (w_u / w_v) \cdot (w'_u / w'_v ) \leq 10$.  With this assumption, we show that given sample access to $w'$ and exact access to $w_v/w_v'$  at the leaves, but only a constant approximate oracle for $w_v/w_v'$ at interior nodes, we can efficiently sample from the leaf distribution of $w$ via a Markov chain that mixes quickly.  This is closely related to the reduction from sampling to weak approximate counting for self-reducible problems in \cite{sj89}.

\begin{theorem}[Sampling a leaf via a Markov chain]\label{lem:sample-from-tree}
Let $(\tree, w), (\tree, w')$ be weighted trees of depth $n$ on the same vertex set such that each vertex has at most $k$ children.
Assume that for any node $u \in \tree$ with child $v \in \tree$, $0.1 \leq (w_u w'_{v})/(w_v w'_{u} ) \leq 10$. Further, assume we are given an oracle that can perform the following types of queries.
\begin{enumerate}
    \item For any internal node $v$, compute an estimate $\wh{r}_v$  such that $0.1 (w_v/w'_v) \leq \wh{r}_v \leq  10 (w_v/w'_v)$. 
    \item For any leaf node $v$, exactly compute $\wh{r}_v = w_v/w'_v$.
    \item For any internal node $v$, respond to a sample query for $(\tree,w')$ (see \cref{def:sampling-sub-tree}). 
\end{enumerate}
Then, for any $0<\eps, \delta<1$, there exists an algorithm that uses $\bigO{ n^4 \log(n k/\eps) \log(1/\delta) }$ queries to the aforementioned oracle and outputs either a leaf node $v \in \tree$ or $\perp$.
The algorithm outputs a leaf with probability $\geq 1 -\delta$, and when the algorithm outputs a leaf, its distribution is $\eps$-close to the leaf distribution of $(\tree,w)$ in TV distance.
\end{theorem}
\begin{proof}
Consider the following random walk on $(\tree,w)$.

\begin{algorithm}[Tree random walk]
    \label{algo:tree} \hfill
\begin{algorithmic}[1]
    \State Set $v$ to be the root node;
    \For{$\bigO{n^3\log(nk/\eps)}$ iterations}
        \State Let $u$ be the parent of $v$, if it exists;
        \State Sample an $\xi \in \{0, 1, 2\}$ with probability $\braces{0.01 \wh{r}_u/\wh{r}_v, 0.01, 0.99 - 0.01 \wh{r}_u/\wh{r}_v}$, respectively.
        \State If $\xi = 0$, set $v$ to $u$;
        \State If $\xi = 1$, set $v$ to a child of $v$ chosen according to a sample query on $(T,w')$, if one exists;
        \State If $\xi = 2$, keep $v$ unchanged;
    \EndFor
    \State \Output $v$;
\end{algorithmic}
\end{algorithm}

Note that for all vertices, we query the oracle once for $\wh{r}_v$ and always use the same estimate throughout the random walk, so that we never query the same vertex again for a new estimate.
By assumption, $\wh{r}_u/\wh{r}_v \leq 10$ so this walk is well-defined.

First, we prove that the stationary distribution of this walk has probability mass on each vertex proportional to $\wh{r}_v w'_v$. In particular, since $\wh{r}_v = \frac{w_v}{w_v'}$ on any leaf $v$, this is exactly the distribution proportional to $w$ on the leaves.

We do this by verifying reversibility (see \cref{def:ergodic-markov-chain}).
Consider any two vertices $u,v$ such that $u$ is a parent of $v$.  Then we have
\[
P_{uv}  \pi_v = 
0.01\frac{\wh{r}_u}{\wh{r}_v} \wh{r}_v w'_v = 0.01\wh{r}_u w'_v = 0.01 \frac{w'_v}{w'_u} \wh{r}_u w'_u = P_{vu} \pi_u
\]
as desired.  Thus we can conclude that the Markov chain is ergodic and time-reversible.

Next, we will lower bound the spectral gap of this walk by bounding the conductance and applying \cref{lem:weighted-cheeger}.  To lower bound the conductance we show that it suffices to consider when the subset $S$ is a sub-tree rooted at some vertex $v$. To see this, consider any cut $(T_1, T_2)$ such that $T_2$ contains the root. Let  $\calM = \braces{ v_i }_{i \in [r]}$ be the maximal elements in $T_1$, i.e.\ for each $v_i$, the parent of $v_i$ , denoted by $u_i$, is in $T_2$. Then, 
\begin{equation*}
    \Phi_{T_1} \geq  \frac{ \sum_{ v_i \in \calM  } P_{v_{i} u_{i} } \pi_{v_{i}}    }{ \sum_{ v_i \in T_1 } \pi_{v_{i} }   } \geq \frac{ \sum_{ v_i \in \calM  } P_{v_{i} u_{i} } \pi_{v_{i}}    }{ \sum_{ v_i \in  \calM } \sum_{v_j \in \textit{ sub-tree } (v_i)
    } \pi_{v_{j} }   }  \geq \frac{ P_{v_{i} u_{i} } \pi_{v_{i}}  }{\sum_{v_j \in \textit{ sub-tree } (v_i)
    } \pi_{v_{j} }  }
\end{equation*}
Let $S$ be the sub-tree rooted at $v$ and let $u$ be the parent.  Then
\[
\Phi_S \geq \frac{0.1 \wh{r}_u w_v'}{\sum_{v' \preceq v} \wh{r}_v w_v'} \geq \frac{w_u w_v'}{40w_u' (\sum_{v' \preceq v} w_v)} \geq \frac{w_u w_v'}{40 n w_u' w_v } \geq \frac{1}{80 n}
\]
where in the above, we used that $\sum_{v' \preceq v} w_v \leq nw_v$ from the definition of a weighted tree of depth $n$.  Thus, by \cref{lem:weighted-cheeger}, the spectral gap of the Markov chain is $\Omega(1/n^2)$.  Finally, the number of nodes in the tree is at most $(k+1)^n$ so after  $O(n^3 \log(k/\eps))$ steps of the Markov chain, the distribution will be $\eps$-close to the stationary distribution (see for instance \cite{lp17}).   Note that the stationary distribution matches the leaf distribution of $(T,w)$ on the leaves and also the probability of being at a leaf in the stationary distribution is at least 
\[
\frac{\sum_{v \text{ leaf}}\wh{r}_v w_v'}{\sum_{v}\wh{r}_v w_v'} \leq \frac{1}{4} \frac{\sum_{v \text{ leaf}} w_v}{\sum_{v}w_v} \geq \frac{1}{4n},
\]
using that the tree is depth $n$.
Thus, with $O(n^3 \log(nk/\eps))$ steps of the Markov chain, we get $\eps/(100 n)$-close to the stationary distribution, and the output will be a leaf with probability $1/(4n)$.
So, by running \cref{algo:tree} until the output is a leaf, with probability $1 - \delta$, we will hit a leaf within $O(n \log(1/\delta))$ epochs and this gives the desired output.
\end{proof}

\subsection{Sample tree}

We want to apply the above algorithm to the sample tree corresponding to \cref{algo:single-step-2}.
In particular, we consider a root node corresponding to an initial input, an empty configuration with zero pinned nodes, $S = [\qubits]$ and $\mathcal{X} = \varnothing$.
Then, the children of this node corespond to calling the algorithm on $(S, \mathcal{X})$, which we can interpret as pinning more unpinned sites and adding them to the configuration.

\begin{definition}[Sample tree of a Hamiltonian]
\label{def:ham-sample-tree}
Given a Hamiltonian $H = \sum_a H_a$ over $n$ sites such that $H$ is a $\locality$-local Hamiltonian with dual interaction graph with degree $\degree$, along with parameters $\beta$ and $\eps$, we construct its sample tree as follows:
\begin{itemize}
\item It has a root node with label $(S, \mathcal{X}) = ([n], \varnothing)$.
\item For each node $v$ such that its label $(S_v, \mathcal{X}_v)$ satisfies $S_v \neq \varnothing$, it has a child for every possible outcome of running \cref{algo:single-step-2} on that node with input $S_v$, $\mathcal{X}_v$, and $\eps$, labeled with the corresponding output $(\wh{S}, \mathcal{\wh{X}})$.
\item For a node $v$, we associate its label $(S_v, \mathcal{X}_v)$ with the matrix
\[
    Q_v = e^{-\frac{\beta}{2} H^{(S_v)}} \sigma(\mathcal{X}_v) e^{-\frac{\beta}{2} H^{(S_v)} }  \,.
\]
\end{itemize}
\end{definition}

Our overall strategy is as follows. 
Assigned to every node in the sample tree is a matrix $Q_v \in \C^{2^n \times 2^n}$; for the root node, this matrix is $e^{-\beta H}$, and for a leaf node, this matrix is a configuration $\sigma(\mathcal{X}_v)$.
In between, we think of $Q_v$ as a ``posterior'' Gibbs state: having pinned all qubits in $[\qubits] \backslash S_v$ according to the (separable) configuration $\sigma(\mathcal{X}_v)$, $Q_v$ is then some Gibbs-like state on $S_v$.

Going from a node to its children pins more sites.
We first show in \cref{lem:sample-tree-property} that the expectation of $Q_{v'}$ over the children $v'$ of a node $v$, with respect to the distribution given by the sampling algorithm \cref{algo:single-step-2}, is close to the matrix $Q_v$.
When $\eps = 0$, this expectation is exact by \cref{lem:sep-expectation}, but we have to take a non-zero $\eps$ to get an efficient algorithm.
Then, by iterating this lemma, we show that the average over the leaves is close to the root matrix, $e^{-\beta H}$.
In other words, iterating \cref{algo:single-step-2} until $S = \varnothing$ gives a configuration $\mathcal{X}$ such that $\sigma(\mathcal{X}) \approx e^{-\beta H}$.

We would like to sample the Gibbs state $e^{-\beta H}/\tr(e^{-\beta H})$ by sampling a leaf $v$ from a certain distribution, and then preparing the corresponding configuration $\sigma(\mathcal{X}_v) / \tr(\sigma(\mathcal{X}_v))$; we define this distribution in \cref{def:true-weight}.
However, this distribution is not the same as the distribution naturally induced by running \cref{algo:single-step-2}; this distribution is defined in \cref{def:natural-weight}.
To prove our main result, we show that the ratio between these two distributions is bounded and can be efficiently approximated.
We then use \cref{lem:sample-from-tree} to sample. 
Throughout this argument, we use the lemmas in \cref{subsec:separability}, which describe the behavior of \cref{algo:single-step-2}, to analyze the sample tree.

We begin by recording a few basic observations about the structure of the sample tree.
In particular, we observe that, throughout, at most one element of $\mathcal{X}_v$ ever intersects with $S_v$, so that $Q_v$ can be written as a tensor product of a set of ``pinned'' sites with a set of ``active'' sites.

\begin{lemma}\label{fact:basic-observations}
The sample tree has the following properties.
For a node $v = (S_v, \mathcal{X}_v)$ at depth $d$:
\begin{enumerate}
    \item $\abs{S_v} \leq \qubits - d$, so the depth of the tree is at most $\qubits$;
    \item For any two children $u,\,u'$ of $v$, $S_{u} = S_{u'}$;
    \item For the root node $u$, $Q_{u} = e^{-\beta H}$, and for a leaf node $v$, $Q_{v} = \sigma(\mathcal{X}_v)$;
    \item For an interior node $v$, $Q_v$ decomposes into a tensor product $Q_v = (Y_v)_{\supp(Y_v)} \otimes (e^{-\frac{\beta}{2} H^{(S_v)}} (\id + c X) e^{-\frac{\beta}{2} H^{(S_v)}})_{[\qubits] \setminus \supp(Y_v)}$, where $(c, X) \in \mathcal{X}_v$ or $(c, X) = (0, \id)$;
    \item When $\beta \leq \frac{1}{200 \degree \locality}$, for every node $v$ and every $(c, X) \in \mathcal{X}_v$, $\abs{c} \leq \frac{1}{2}$; in particular, $Q_v$ is PSD.
\end{enumerate}
\end{lemma}
\begin{proof}
1 and 2 follow from inspection of \cref{algo:single-step-2}; it always removes elements from $S_v$, and it makes this choice deterministically.
3 follows from the definition of $Q_v$.
4 follows from \cref{lem:well-defined}.
5 follows from \cref{lem:coeff-potential}.
\end{proof}

We will also need a bound on the running time of \cref{algo:single-step-2}.

\begin{lemma}[Algorithm complexity]\label{lem:one-step-runtime}
    Let $H = \sum_{a \in [\terms]} H_a$ be a Hamiltonian and we assume access to it as in \cref{rmk:ham-input}.
    Then \cref{algo:single-step-2} can be implemented to run in time $\bigOt{\log(\qubits/\eps) \locality \degree}$, where the output configuration is described in terms of the descriptions of its Hermition monomials.
    Further, the number of possible outputs is $\bigO{(40\log(\qubits/\eps)(\degree + 1))^{20 \log(\qubits/\eps)}}$.
\end{lemma}
\begin{proof}
We begin by considering \cref{algo:recursive-sample-term} with input $H$, $k \geq 0$, and $\abs{\mathcal{Q}} \leq \degree + 1$.
Then this algorithm can be run in $\bigOt{k \degree}$ time.
The main technical consideration is maintaining $\mathcal{R}_t$, which can be done by adding the neighbors of the term being added to $\mathcal{R}_t$ to form $\mathcal{R}_{t+1}$.
This costs $\bigOt{\degree}$ time since there are at most that many neighbors.
Further, the number of possible outputs $c_k, b^{(k)}$ is $\bigO{((k-1)!) \cdot (4(\degree + 1))^k}$, since $\abs{\mathcal{R}_t} \leq t(\degree + 1)$.

Next, consider \cref{algo:sample-term}, which calls \cref{algo:recursive-sample-term} for some $k$ between $1$ and $t_{\max}$.
So, its running time is $\bigOt{t_{\max} \degree}$ and the number of possible outputs is $\bigO{(t_{\max}!) \cdot (4(\degree + 1))^{t_{\max}}}$.

Finally, consider \cref{algo:single-step-2}.
This algorithm can be run in $\bigOt{t_{\max} \degree \locality}$ time: it can be seen as taking an element of the configuration $(c_{\wh{\ell}}, X_{\wh{\ell}}) \in \mathcal{X}$ and changing it to have at most $2 t_{\max}$ more terms in it.
The step of finding an element of $\mathcal{E}^{(S)}$ neighboring a term in $X_{\wh{\ell}}$ can be done by maintaining a list of such terms; then, when $X_{\wh{\ell}}$ is changed, this list can be updated by removing elements when $S$ is modified and adding elements corresponding to the terms being added.
The algorithm is deterministic apart from the two calls to \cref{algo:sample-term} and the choice of $\xi$ among seven options.
We know that $\abs{\mathcal{Q}} \leq \degree + 1$, so the number of possible outputs is
\begin{align*}
    \bigO{((t_{\max}!) \cdot (4(\degree + 1))^{t_{\max}})^2}
    = \bigO{(40\log(\qubits/\eps)(\degree + 1))^{20 \log(\qubits/\eps)}}.
\end{align*}
Note that, across all of these algorithms, we never use the matrix representations of the Hamiltonian terms, instead only using the information coming from the dual interaction graph and the support of every term.
\end{proof}

\subsection{Weight function analysis}

Now we define two different weighted trees on the sample tree.  The first is the weight derived from the sampling process.

\begin{definition}[Natural weight]\label{def:natural-weight}
Given a Hamiltonian $H = \sum_a H_a$ over $n$ sites and parameters $\beta, \locality, \degree$ such that $H$ is a $\locality$-local Hamiltonian with degree $\degree$, let $\mathcal{T}$ be its sample tree.
We define the natural weight function $\omega$ to have for each node $v \in \mathcal{T}$, $\omega(v)$ is the probability of reaching $v$ from the root by running the sampling process in \cref{algo:single-step-2} at each intermediate node.
\end{definition}

\begin{remark}
The weight at any intermediate node is equal to the sum of the weights of the leaves of its sub-tree, so $\omega$ indeed defines a valid weighted tree.
\end{remark}

The second weight function, the true weight, is defined by adjusting the natural weight at each leaf by $\tr(\sigma(\mathcal{X}))$.  We will show in \cref{coro:true-weight-error} that to sample from the Gibbs state, it suffices to sample a leaf according to this true weight distribution.

\begin{definition}[True weight]\label{def:true-weight}
Given a Hamiltonian $H = \sum_a H_a$ over $n$ sites and parameters $\beta, \locality, \degree$ such that $H$ is a $\locality$-local Hamiltonian with degree $\degree$, let $T$ be its sample tree.  We define the true weight function $\kappa$ as follows: for each leaf node $v$, we set $\kappa(v) = \tr(\sigma(\mathcal{X}_v)) \omega(v)$ and for each intermediate node $v$, $\kappa(v)$ is the sum of the weights of the leaves in its sub-tree.
\end{definition}
\begin{remark}
By \cref{fact:basic-observations}, when $\beta \leq 1/(200 \degree \locality)$, all of the $\sigma(\mathcal{X}_v)$ are PSD and thus this is a valid weighted tree.  
\end{remark}

First, we prove that $Q_v$ is close to the average of $Q_{v'}$ over its children $v'$, according to the natural weight.
Since these matrices are not trace-normalized, it will be important to ensure that our error bounds are ``at the right scale''.
We do this by bounding our errors multiplicatively in PSD ordering.

\begin{lemma}\label{lem:sample-tree-property}
Given a Hamiltonian $H = \sum_a H_a$ over $n$ sites and parameters $\beta, \locality, \degree$ such that $H$ is a $\locality$-local Hamiltonian with degree $\degree$ and $\beta \leq 1/(200 \degree \locality)$, let $\tree$ be its sample tree as defined in \cref{def:ham-sample-tree}.  Let $v \in \tree$ be a node in the tree.  Then we have
\[
    \parens[\Big]{1 - \frac{\eps}{20 n}} Q_v \preceq \sum_{v' \text{ child of } v} \frac{\omega(v')}{\omega(v)} Q_{v'} \preceq \parens[\Big]{1 + \frac{\eps}{20 n}} Q_v.
\]
\end{lemma}
\begin{proof}
Consider the execution of \cref{algo:single-step-2} with input $(S_v, \mathcal{X}_v)$.
Then by \cref{lem:sep-expectation}, over the expectation of the children $v'$ of $v$,
\begin{align*}
    \E_{v'}[\sigma(\mathcal{X}_{v'})] = T_{t_{\max}, \beta/2}(H^{(S_v)}, H_{(S_v \setminus S_{v'})}^{(S_v)})^\dagger \sigma(\mathcal{X}_v) T_{t_{\max}, \beta/2}(H^{(S_v)}, H_{(S_v \setminus S_{v'})}^{(S_v)}) \, .
\end{align*}
Here, we are using \cref{fact:basic-observations}, that every child $v'$ of $v$ has the same $S_{v'}$.
By applying \cref{claim:error-bound-2}, we have that
\begin{align*}
    \E_{v'}[\sigma(\mathcal{X}_{v'})] \preceq \parens[\Big]{1 + \frac{100}{C^{t_{\max}}}} e^{\frac{\beta}{2}(H^{(S_v)} - H_{(S_v \setminus S_{v'})}^{(S_v)})} e^{-\frac{\beta}{2}H^{(S_v)}} \sigma(\mathcal{X}_v) e^{-\frac{\beta}{2}H^{(S_v)}} e^{\frac{\beta}{2}(H^{(S_v)} - H_{(S_v \setminus S_{v'})}^{(S_v)})}
\end{align*}
Using that $t_{\max} = 10 \log(\qubits / \eps)$ and left and right multiplying the above by $e^{-\frac{\beta}{2} (H^{(S_v )}  - H_{(S_v \setminus S_{v'})}^{(S_v )})}$ and applying \cref{claim:error-bound-2}, we get
\begin{align*}
    \E_{v'}[Q_{v'}] \preceq \parens[\Big]{1 + \frac{\eps}{10\qubits}} Q_v.
\end{align*}
Here, we use that $H^{(S_v)} - H_{(S_v \setminus S_{v'})}^{(S_v)} = H^{(S_{v'})}$.
This gives us the upper bound, since averaging over the children of $v$ with weights $\omega(v')/\omega(v)$ is exactly the same as taking the expectation over the execution of \cref{algo:single-step-2}: $\E_{v'}[Q_{v'}] = \sum_{v' \text{ child of } v} \frac{\omega(v')}{\omega(v)} Q_{v'}$.
By applying the other side of \cref{claim:error-bound-2}, we get the corresponding lower bound.
\end{proof}

By iterating \cref{lem:sample-tree-property}, we can relate the average at the leaves to the Gibbs state. 
Specifically, since the true weight distribution on the leaves is exactly defined to be equal to the natural weight distribution arising from the sampling process and then distorted by a factor of $\alpha_v$ at each leaf $v$, we get that the average of the (normalized) states at the leaves according to the true weight distribution is close to the Gibbs state $e^{-\beta H}/\tr(e^{-\beta H})$.

\begin{corollary}[Average of the leaves is close to the Gibbs state]\label{coro:true-weight-error}
Given a Hamiltonian $H = \sum_a H_a$ over $n$ sites and parameters $\beta, \locality, \degree$ such that $H$ is a $\locality$-local Hamiltonian with degree $\degree$ and $\beta \leq 1/(200 \degree \locality )$, let $\tree$ be its sample tree.
Then 
\[
    \norm[\Big]{ \frac{e^{-\beta H}}{\tr(e^{-\beta H})} - \sum_{v \text{ leaf of } \tree} \frac{\kappa(v)}{\sum_{v' \text{ leaf of } \tree } \kappa(v')} \frac{\sigma(\mathcal{X}_v)}{\tr(\sigma(\mathcal{X}_v))}}_1  \leq \frac{\eps}{2} \,.
\]
\end{corollary}
\begin{proof}
By repeatedly applying \cref{lem:sample-tree-property} starting from the root to the leaves, we have
\[
    \parens[\Big]{1 - \frac{\eps}{10}} e^{-\beta H}
    \preceq \sum_{v \text{ leaf of } \tree} \kappa(v) \frac{\sigma(\mathcal{X}_v)}{\tr(\sigma(\mathcal{X}_v))}
    \preceq \parens[\Big]{1 +\frac{\eps}{10}} e^{-\beta H} \,,
\]
which follows from recalling that $\kappa(v) = \tr(\sigma(\mathcal{X}_v))\omega(v)$.
Thus
\begin{align*}
    & \norm[\Big]{ \frac{e^{-\beta H}}{\tr(e^{-\beta H})} - \sum_{v \text{ leaf of } \tree} \frac{\kappa(v)}{\sum_{v' \text{ leaf of } \tree } \kappa(v')} \frac{\sigma(\mathcal{X}_v)}{\tr(\sigma(\mathcal{X}_v))}}_1 \\
    &\leq 2\norm[\Big]{ \frac{e^{-\beta H}}{\tr(e^{-\beta H})} - \sum_{v \text{ leaf of } \tree} \frac{\kappa(v)}{\tr(e^{-\beta H})} \frac{\sigma(\mathcal{X}_v)}{\tr(\sigma(\mathcal{X}_v))}}_1
    \leq \frac{\eps}{2}
\end{align*}
as desired.
\end{proof}

To sample from the distribution induced by $\kappa$ on the leaves, we will need to approximate the ratio $\kappa(v)/\omega(v)$ in order to then apply \cref{lem:sample-from-tree} (since we can run \cref{algo:single-step-2} to sample from the weighted tree $\omega$).  We do this below.

\begin{corollary}[Bounded weight ratio]\label{coro:weight-ratio}
Given a Hamiltonian $H = \sum_a H_a$ over $n$ sites and parameters $\beta, \locality, \degree$ such that $H$ is a $\locality$-local Hamiltonian with degree $\degree$ and  $\beta \leq 1/(200 \degree \locality )$, let $\tree$ be its sample tree.
Let $v$ be a node with associated label $(S_v, \mathcal{X}_v)$.
Then we have
\[
0.9 \tr(Q_v) \leq \frac{\kappa(v)}{\omega(v)} \leq 1.1 \tr(Q_v) \,.
\]
\end{corollary}
\begin{proof}
By repeatedly applying \cref{lem:sample-tree-property}, we have
\[
    \left(1 - \frac{\eps}{20 n} \right)^{\abs{S_v}} Q_v \preceq \sum_{\substack{v' \text{ a leaf } \\ v' \text{ descendant of } v}} \frac{\omega(v')}{\omega(v)} \sigma(\mathcal{X}_{v'}) \preceq \left(1 + \frac{\eps}{20 n} \right)^{\abs{S_v}} Q_v \,.
\]
Now take the trace of the above to get the desired bound:
\[
    \left(1 - \frac{\eps}{20 n} \right)^{\abs{S_v}} \tr(Q_v) \leq \frac{\kappa(v)}{\omega(v)} \leq \left(1 + \frac{\eps}{20 n} \right)^{\abs{S_v}} \tr(Q_v) \qedhere
\]
\end{proof}

Now we can put everything together to prove our our main theorem, \cref{thm:main-gibbs}.

\begin{proof}[Proof of \cref{thm:main-gibbs}]
We will apply \cref{lem:sample-from-tree} on the sample-tree with $w' \leftarrow \omega$ and $w \leftarrow \kappa$.
First, we verify the hypotheses of \cref{lem:sample-from-tree}.
By \cref{lem:one-step-runtime}, the number of children of each node is at most $k = \bigO{(40\log(\qubits/\eps)(\degree + 1))^{20\log(\qubits/\eps)}}$.
Consider a node $u$ with child $v$, and consider their associated labels $(S_u, \mathcal{X}_u)$ and $(S_v, \mathcal{X}_v)$.
Note this means $S_v = S_u \backslash R$ for some $R = \supp(H_a)$.
By \cref{coro:weight-ratio}, 
\[
    0.8 \frac{\tr(Q_u)}{\tr(Q_v)} \leq \frac{w_u w_v'}{w_v w_u'} \leq 1.25 \frac{\tr(Q_u)}{\tr(Q_v)}.
\]
Recall that $Q_u = e^{-\frac{\beta}{2} H^{(S_u)}} \sigma(\mathcal{X}_u) e^{-\frac{\beta}{2} H^{(S_u)}}$, and respectively for $Q_v$.
$(S_v, \mathcal{X}_v)$ is a possible outcome of running \cref{algo:single-step-2} on $(S_u, \mathcal{X}_u)$.
So, because of \cref{lem:well-defined}, we can write $Q_u = Y_{\supp(Y)} \otimes (e^{-\frac{\beta}{2} H^{(S_u)}} (\id + cX) e^{-\frac{\beta}{2} H^{(S_u)}})_{[\qubits] \setminus \supp(Y)}$ and $Q_v = Y_{\supp(Y)} \otimes (e^{-\frac{\beta}{2} H^{(S_v)}} (\id + \wh{c}\wh{X}) e^{-\frac{\beta}{2} H^{(S_v)}})_{[\qubits] \setminus \supp(Y)}$ for $(c, X)$ and $(\wh{c}, \wh{X})$ corresponding to the $(c_{\wh{\ell}}, X_{\wh{\ell}})$ and $(\wh{c}, \wh{X})$ in \cref{algo:single-step-2}.
Here, we are treating the unchanged part of the configuration of $\mathcal{X}_u$ as $Y$.
\begin{align*}
    \frac{\tr(Q_u)}{\tr(Q_v)}
    = \frac{\tr(e^{-\frac{\beta}{2} H^{(S_u)}} (\id + cX) e^{-\frac{\beta}{2} H^{(S_u)}})}{\tr(e^{-\frac{\beta}{2} H^{(S_v)}} (\id + \wh{c}\wh{X}) e^{-\frac{\beta}{2} H^{(S_v)}})}
    \leq 3 \frac{\tr(e^{-\beta H^{(S_u)}})}{\tr(e^{-\beta H^{(S_v)}})}
    \leq 4.5
\end{align*}
The first inequality follows from \cref{fact:basic-observations}, which implies that $\abs{c}, \abs{\wh{c}} \leq 1/2$, so $0.5\id \preceq \id + cX \preceq 1.5 \id$.
The second follows from \cref{claim:psd-ordering} which implies that $\frac{3}{4} \leq \tr(e^{-\beta H^{(S_u)}}) / \tr(e^{-\beta H^{(S_v)}}) \leq \frac{3}{2}$.
Combining this with the analogous lower bound, we deduce
\[
0.1 \leq \frac{w_u w_v'}{w_vw_u'} \leq 10 \,.
\]
Next, we show how to implement the types of queries required in \cref{lem:sample-from-tree}.
For the first type of query, estimating $w_v / w_v'$ to multiplicative error, by \cref{coro:weight-ratio} it suffices to be able to estimate $\tr(Q_v)$.
By \cref{lem:sample-tree-property}, we can write $\tr(Q_v) = \tr((Y_v)_{\supp(Y_v)}) \tr((e^{-\frac{\beta}{2} H^{(S_v)}} (\id + cX) e^{-\frac{\beta}{2} H^{(S_v)}})_{[\qubits] \setminus \supp(Y_v)})$.
Computing the trace associated with $Y_v$ can be done exactly, as $Y_v$ is a tensor product of Hermitian monomials from the configuration $\mathcal{X}_v$.
Using that $\abs{c} \leq 1/2$, the second trace is between $0.5 \tr((e^{-\beta H^{(S_v)}})_{[\qubits] \setminus \supp(Y_v)})$ and $1.5 \tr((e^{-\beta H^{(S_v)}})_{[\qubits] \setminus \supp(Y_v)})$, and this can be estimated via \cref{thm:estimation-log-partition} with $\eta \leftarrow 0.01$.  The running time of answering this query is
\[
O\left( \qubits \cdot (100\qubits)^{ \frac{ 4+\log(\degree)} {\log(\beta_c/\beta) } }  \cdot \locality \cdot \degree^2 \cdot \polylog(\qubits)\right)
\]
For the third type of queries, we simply run \cref{algo:single-step-2}.
By \cref{lem:one-step-runtime}, the running time is $\bigOt{\log(n/\eps) \poly(\locality, \degree)}$.
Note that when running the Markov chain in \cref{lem:sample-from-tree}, we can store the labels $(S_u, \mathcal{X}_u)$ of all of the nodes that we have visited so far.
When visiting a new node, computing the new label $(S, \mathcal{X})$ involves just one execution of \cref{algo:single-step-2} on its parent, which we must have already visited.
Thus, whenever we visit a leaf $v$, can compute $\tr(\sigma(\mathcal{X}_v)) = \kappa(v)/\omega(v)$ exactly, allowing us to answer the second type of query whenever we need to (which is only when the Markov chain visits a leaf).
Thus, we have verified all of the hypotheses of \cref{lem:sample-from-tree}.
Putting everything together, we get that with probability $1 - \delta$, we get a sample from a distribution that is $\eps/4$-close in TV to the leaf distribution of $\kappa$ in time 
\[ 
\bigOt*{(100\qubits)^{ 5 +\frac{  \cdot (4+\log(\degree))}{\log(\beta_c/\beta)} } \log^2(1/\eps) \log(1/\delta) \poly(\locality, \degree) } \,.
\]
By setting $\delta \gets \eps/4$, we get that the output distribution including failure is $\eps/2$-close in TV distance to the leaf distribution of $\kappa$.

A leaf $v$ corresponds to a configuration $\mathcal{X}_v$; by applying \cref{lem:pauli-monomial} to each $(c, X) \in \mathcal{X}_v$, in $\bigOt{\qubits \log(1/\eps)\locality}$ time we can prepare a corresponding stabilizer product state $\ket{\psi}$ such that $\E[\ketbra{\psi}{\psi}] = \sigma(\mathcal{X}_v) / \tr(\sigma(\mathcal{X}_v))$.
Note that, here we use that the sum of the degrees of the monomials involved in $\mathcal{X}_v$ is bounded by $\bigO{\qubits t_{\max}}$, the number of iterations multiplied by the largest the monomial degree can increase per-iteration.
Finally,
\begin{multline*}
    \norm[\Big]{ \frac{e^{-\beta H}}{\tr(e^{-\beta H})} 
    - \E[\ketbra{\psi}{\psi}]}_1
    \leq \norm[\Big]{\frac{e^{-\beta H}}{\tr(e^{-\beta H})}  -  \sum_{v \text{ leaf of } T} \frac{ \kappa(v)}{\sum_{v' \text{ leaf of } T} \kappa(v')} \frac{\sigma(\mathcal{X}_{v})}{\tr(\sigma(\mathcal{X}_{v'}))}}_1 \\
    + \norm[\Big]{\sum_{v \text{ leaf of } T} \frac{ \kappa(v)}{\sum_{v' \text{ leaf of } T} \kappa(v')} \frac{\sigma(\mathcal{X}_{v})}{\tr(\sigma(\mathcal{X}_{v'}))} - \E\bracks[\Big]{\frac{\sigma(\mathcal{X}_v)}{\tr(\sigma(\mathcal{X}_v))}}}_1
    \leq \eps .
\end{multline*}
where the expectation is over the execution of the Markov chain in \cref{lem:sample-from-tree}, conditioned on a valid output, and the call to \cref{lem:pauli-monomial}.
Note that in the last step above we used \cref{coro:true-weight-error} and the fact that the distribution of our output is $\eps/2$ close to the distribution induced by $\kappa$ on the leaves.
Thus, we can simply output $\wh{\rho} = \ketbra{\psi}{\psi}$ and this completes the proof.
The bound we required on the temperature is that $\beta \leq 1/(200 \degree \locality)$; we set $\beta_c$ sufficiently small so that the exponent in the running time can be bounded $5 + \frac{4 + \log(\degree)}{\log(\beta_c / \beta)} \leq 6 + \frac{\log(\degree)}{\log(\beta_c / \beta)}$, and that the $100^{\frac{\log(\degree)}{\log(\beta_c / \beta)}}$ can be folded into the $\poly(\degree)$.
\end{proof}

\section*{Acknowledgments}
\addcontentsline{toc}{section}{Acknowledgments}

The authors thank Nikhil Srivastava, Álvaro Alhambra, Cambyse Rouzé, Daniel Stilck França, and Chi-Fang Chen for enlightening discussions.
In particular, we thank Álvaro for highlighting the structural implications of our work and for drawing a connection to sudden death of entanglement.

AB is supported by Ankur Moitra's ONR grant and the NSF TRIPODS program (award DMS-2022448). AL is supported in part by an NSF GRFP and a Hertz Fellowship.  AM is supported in part by a Microsoft Trustworthy AI Grant, an ONR grant and a David and Lucile Packard Fellowship.  ET is supported by the Miller Institute for Basic Research in Science, University of California Berkeley. 
This work was done in part while the authors were at the Simons Institute for the Theory of Computing.

\printbibliography

\appendix

\section{Cluster expansion}
\label{appendix-cluster-exp}

We give bounds on the Taylor series expansion of the log-partition function of a local Hamiltonian, closely following the cluster expansion argument of Mann and Helmuth~\cite{mh21,kp86}.
We reproduce this for clarity, increased generality, and completeness, when merged this with results on computing the cluster expansion \cite{hkt21}.
We work with the dual interaction graph, as opposed to the interaction hypergraph as in \cite{mh21}, to give slightly tighter bounds.
The arguments here straightforwardly extend to qudit systems, but for the sake of simplicity we present only the qubit setting.

\paragraph{Notation.}
For this section, we will be working with multisets.
For a base set $S$, a multiset of $S$, $\cluster T \colon S \to \mathbb{N}$, is defined in terms of its multiplicity function.
The support of $\cluster T$ is the set of $s \in S$ with $\cluster T(s) > 0$, and the size of $\cluster T$ is denoted $\abs{\cluster T} = \sum_{s \in S} \cluster T(s)$.
We can take unions of multisets, $\cluster T_1 \cup \cluster T_2$, in the usual way, where its multiplicity function is the sum of the two individual multiplicities.
We use the factorial notation $\cluster{T}! = \prod_{s \in S} (\cluster{T}(s)!)$.
For a collection of complex numbers indexed by $S$, $\{x_s\}_{s \in S}$, $x^{\cluster T} = \prod_{s \in S} x_s^{\cluster T(s)}$ is the product of variables associated with $\cluster T$ with multiplicity.

We use the Iverson bracket: for a proposition $P$, $\iver{P}$ equals one when $P$ is true, and zero otherwise.
We denote the normalized trace by $\ntr(A) = \tr(A) / \tr(\id)$.
The symmetric group with order $\qubits!$ is denoted $\sym_\qubits$.

\subsection{Abstract polymer model}

We first describe cluster expansion for abstract polymer models, along with the Kotecký--Preiss condition for convergence of the expansion~\cite{kp86}.
This is fully standard; we follow the formulation of Friedli and Velenik~\cite{fv17}, which in turn follows the analysis of Ueltschi~\cite{ueltschi03}.

\begin{definition}[Abstract polymer model]
    An \emph{abstract polymer model} is a collection of objects $K$ which we refer to as \emph{polymers}, along with associated weights $\{w_\gamma\}_{\gamma \in K}$ in $\C$ and a symmetric relation $\sim$ on $K$ such that, for all $\gamma \in K$, $\gamma \nsim \gamma$.
    Here, $\sim$ denotes compatibility: for example, $\gamma \nsim \gamma$ means that a polymer is incompatible with itself.

    For an ordered list of polymers $\Gamma = (\gamma_1,\dots,\gamma_\ell)$, its \emph{incompatibility graph} $H_\Gamma$ has $\ell$ vertices indexed by $\Gamma$ and an edge between $\gamma_i$ and $\gamma_j$ if and only if $\gamma_i \nsim \gamma_j$.
    $\Gamma$ is a \emph{cluster} if its incompatibility graph is connected.\footnote{
        The definition of a cluster is different from \cite{hkt21}, but this is the definition that is consistent with the rest of the cluster expansion literature.
    }
\end{definition}

\begin{definition}[Partition function of a polymer model, {\cite[Definition 5.2]{fv17}}] \label{def:fv-part}
    A polymer model has an associated \emph{partition function}.
    \begin{align*}
        Z = \sum_{S \subseteq K} \parens[\Big]{\prod_{\gamma \in S} w_\gamma}\iver{S \text{ is pairwise compatible}}.
    \end{align*}
    We say that a subset $S \subseteq K$ is pairwise compatible if, for all distinct $\gamma, \gamma' \in S$, $\gamma \sim \gamma'$.
\end{definition}

\begin{definition}[Ursell function]
    The \emph{Ursell function} $\varphi$ is a function mapping graphs to real numbers.
    For $H = (V, E)$, 
    \begin{align*}
        \varphi(H) = \frac{1}{\abs{V}!}\sum_{\substack{A \subseteq E \\ \text{spanning} \\ \text{connected} }} (-1)^{\abs{A}}
    \end{align*}
\end{definition}

\begin{proposition}[Formal expansion for the log-partition function, {\cite[Proposition 5.3]{fv17}}] \label{prop:fv-series-defn}
    The following formal equality holds.
    \begin{align*}
        \log Z &= \sum_{k = 1}^{\infty} \sum_{(\gamma_1,\dots,\gamma_k) \in K^{k}} \varphi(H_{(\gamma_1,\dots,\gamma_k)}) \prod_{i \in [k]} w_{\gamma_i}
    \end{align*}
\end{proposition}

\begin{theorem}[Kotecký--Preiss convergence condition, {\cite[Theorem 5.4]{fv17}}] \label{thm:kp-convergence}
    Suppose that there exists a function $\xi \colon K \to \R_{> 0}$ and weights $\{\wh{w}_\gamma\}_{\gamma \in K}$ such that, for each $\gamma_{*} \in K$,
    \begin{align}
        \sum_{\gamma \in K} \abs{\wh{w}_\gamma} e^{\xi(\gamma)} \iver{\gamma \nsim \gamma_{*}} \leq \xi(\gamma_*).
    \end{align}
    Then, for all $\gamma_1 \in K$,
    \begin{align}
        1 + \sum_{k \geq 2} k \sum_{\gamma_2, \dots, \gamma_k} \abs{\varphi(H_{\gamma_1,\dots,\gamma_k})} \prod_{j=2}^k \abs{\wh{w}_{\gamma_j}} \leq e^{\xi(\gamma_1)}.
    \end{align}
\end{theorem}

If $\wh{w} = w$, then this implies convergence of the cluster expansion series (provided the set of polymers is finite).
We will take $\wh{w}_\gamma = w_\gamma e^{b(\gamma)}$ for some function $b$ to get a stronger bound on the convergence of the series.

\subsection{Quantum spin system as an abstract polymer model}

This section follows the work of Mann and Helmuth~\cite{mh21}, which in turn follows Neto\v{c}n\'{y} and Redig~\cite{nr04}.

\begin{lemma}[{Following \cite[Lemma 2]{mh21}}] \label{lem:mh2}
    Let $H = \sum_{a=1}^\terms \lambda_a E_a$ be a low-intersection Hamiltonian with dual interaction graph $\graph$ (see \cref{def:low-intersection-ham}).
    Then $\ntr \exp(-\beta H)$ equals the partition function of the polymer model where $K$ is the space of non-empty multisets of $[\terms]$ whose support is connected on $\graph$; $\gamma \sim \gamma'$ if and only if their distance on $\graph$ is at least two; and the weight function is, for a polymer $\gamma$ with $\abs{\gamma}$ elements,
    \begin{align*}
        w_\gamma = \frac{(-\beta)^{\abs{\gamma}}}{\abs{\gamma}! \gamma!} \ntr\parens[\Big]{\sum_{\sigma \in \sym_{\abs{\gamma}}} \prod_{i=1}^{\abs{\gamma}} E_{\gamma_{\sigma(i)}}} \lambda^\gamma
    \end{align*}
\end{lemma}
Combining this with the formal equality in \cref{prop:fv-series-defn}, we deduce the following expression for the log-partition function of the Hamiltonian,
\begin{align*}
    & \log \ntr \exp(-\beta H) \\
    &= \sum_{k = 1}^{\infty} \sum_{(\gamma_1,\dots,\gamma_k) \in K^{k}} \varphi(H_{(\gamma_1,\dots,\gamma_k)}) \prod_{i \in [k]} w_{\gamma_i} \\
    &= \sum_{k = 1}^{\infty} \sum_{(\gamma_1,\dots,\gamma_k) \in K^{k}} \varphi(H_{(\gamma_1,\dots,\gamma_k)}) \prod_{i \in [k]} \parens[\Big]{\frac{(-\beta)^{\abs{\gamma_i}}}{\abs{\gamma_i}! \gamma_i!} \ntr\parens[\Big]{\sum_{\sigma \in \sym_{\abs{\gamma_i}}} \prod_{j=1}^{\abs{\gamma_i}} E_{(\gamma_i)_{\sigma(j)}}} \lambda^{\gamma_i}} \\
    &= \sum_{k = 1}^{\infty} \sum_{(\gamma_1,\dots,\gamma_k) \in K^{k}} (-\beta \lambda)^{\gamma_1 \cup \dots \cup \gamma_k} \varphi(H_{(\gamma_1,\dots,\gamma_k)}) \prod_{i \in [k]} \parens[\Big]{\frac{1}{\abs{\gamma_i}! \gamma_i!} \ntr\parens[\Big]{\sum_{\sigma \in \sym_{\abs{\gamma_i}}} \prod_{j=1}^{\abs{\gamma_i}} E_{(\gamma_i)_{\sigma(j)}}}}.
\end{align*}
This cluster expansion series is closely related to the multivariate Taylor series expansion of $\log \ntr \exp(-\beta H)$: if we group together all the terms with the same monomial $\lambda^{\gamma}$, we get the Taylor series expansion around $\lambda = (0,\dots,0)$.

\begin{proof}[Proof sketch.]
We can write
\begin{align*}
    \ntr \exp(-\beta H)
    &= \sum_{k \geq 0} \frac{1}{k!}\ntr((-\beta H)^k) \\
    &= \sum_{k \geq 0} \frac{(-\beta)^k}{k!} \sum_{a_1,\dots,a_k} \ntr( E_{a_1}\dots E_{a_k}) \lambda_{a_1}\dots\lambda_{a_k}.
\end{align*}
We want to show that this is equal to the expression for the partition function from \cref{def:fv-part}:
\begin{align*}
    \sum_{S \subseteq K} \iver{S \text{ is pairwise compatible}}\prod_{\gamma \in S} \frac{(-\beta)^{\abs{\gamma}}}{\abs{\gamma}! \gamma!} \ntr\parens[\Big]{\sum_{\sigma \in \sym_{\abs{\gamma}}} \prod_{i=1}^{\abs{\gamma}} E_{\gamma_{\sigma(i)}}} \lambda^\gamma.
\end{align*}
Since $\ntr(XY) = \ntr(X)\ntr(Y)$ when $X$ and $Y$ have disjoint supports, we can split up $\ntr(E_{a_1}\dots E_{a_k})$ into a product of vertex-disjoint parts, each of which corresponds to a polymer $\gamma$.
Because these parts are vertex-disjoint, their distance on $\graph$ is at least two, and so these polymers are compatible.
There is still an underlying ordering $a_1,\dots,a_k$, but after grouping terms appropriately and counting, we can get a sum over ``unordered'' polymers which take the desired form.
\end{proof}

Quantum many-body systems satisfy the Kotecký--Preiss condition at high temperature.\footnote{
    We note here a small discrepancy in the proof of \cite[Lemma 4]{mh21}: there, they use that the number of vertices in a polymer $\gamma$ is at most $\abs{\gamma} + 1$, but this fact is only true for two-local Hamiltonians.
    The analysis presented here accounts for this, getting a critical temperature which is slightly weaker (by a factor of about $\locality$) than what is claimed in \cite{mh21}.
}

\begin{fact} \label{fact:cluster-count}
By \cite[Eqs.\ 29, 36]{hkt21}, there are at most
\begin{align*}
    & \sum_{v=1}^w \binom{v(\degree - 1) + 1}{v - 1} \frac{\degree}{v(\degree - 1) + 1}\binom{w-1}{v-1} \\
    & = \sum_{v=1}^w \binom{v(\degree - 1)}{v - 1} \frac{\degree}{v(\degree - 2) + 3}\binom{w-1}{v-1} \\
    & \leq \sum_{v=1}^w \binom{v(\degree - 1)}{v - 1} \binom{w-1}{v-1} \\
    & \leq e \sum_{v=1}^w (e(\degree - 1))^{v-1} \binom{w-1}{v-1} \\
    & = e (1 + e(\degree - 1))^{w-1}
\end{align*}
clusters $\gamma$ such that $a \in \gamma$ and $\abs{\gamma} = w$.
\end{fact}

\begin{lemma}[{\cite[Lemma 4]{mh21}}] \label{lem:mh4}
    The polymer model in \cref{lem:mh2} satisfies the condition in \cref{thm:kp-convergence} for $\wh{w}_\gamma = w_\gamma (\frac{\beta_c}{\beta})^{\abs{\gamma}}$ and $\xi(\gamma) = \abs{\gamma}$, provided $\beta < \beta_c = \frac{1}{e(e+1)(1 + e(\degree - 1))}$:
    \begin{align} \label{eq:kp-assumption}
        \sum_{\gamma \in K} \abs{\wh{w}_\gamma} e^{\xi(\gamma)} \iver{\gamma \nsim \gamma_{*}} \leq \xi(\gamma_*).
    \end{align}
    This means that, for all $\gamma_1 \in K$,
    \begin{align} \label{eq:kp-conclusion}
        1 + \sum_{k \geq 2} k \sum_{\gamma_2, \dots, \gamma_k} \abs{\varphi(H_{\gamma_1,\dots,\gamma_k})} \prod_{j=2}^k \abs{w_{\gamma_j}}\parens[\big]{\frac{\beta_c}{\beta}}^{\abs{\gamma_j}} \leq e^{\abs{\gamma_1}}.
    \end{align}
\end{lemma}
\begin{proof}
First, we bound the weight $w_\gamma$ of a polymer $\gamma$.
\begin{align} \label{eq:polymer-weight-bound}
    \abs{w_\gamma} &= \abs[\Big]{\frac{(-\beta)^{\abs{\gamma}}}{\abs{\gamma}! \gamma!} \ntr\parens[\Big]{\sum_{\sigma \in S_{\abs{\gamma}}} \prod_{i=1}^{\abs{\gamma}} E_{\gamma_{\sigma(i)}}} \lambda^\gamma}
    \leq \frac{\beta^{\abs{\gamma}}}{\abs{\gamma}! \gamma!} \abs{\gamma}! \abs{\lambda^\gamma}
    \leq \frac{\beta^{\abs{\gamma}}}{\gamma!}
    \leq \beta^{\abs{\gamma}}
\end{align}
This bound implies that $\abs{\wh{w}_\gamma} = \abs{w_\gamma}(\frac{\beta_c}{\beta})^{\abs{\gamma}} \leq \beta_c^{\abs{\gamma}}$.
Now, we bound the expression in the lemma statement.
First consider a polymer $\gamma_{\{a\}}$ which corresponds to the multiset containing a single term $a \in [\terms]$.
\begin{align*}
    \sum_{\gamma \in K} \abs{\wh{w}_\gamma} e^{\xi(\gamma)} \iver{\gamma \nsim \gamma_{\{a\}}}
    &\leq \sum_{\gamma \in K} (e\beta_c)^{\abs{\gamma}} \iver{\gamma \nsim \gamma_{\{a\}}}
\end{align*}
The set of polymers which are incompatible with $\gamma_{\{a\}}$ are the set of connected multisets of $[\terms]$ which are adjacent to $a$.
By \cref{fact:cluster-count}, there are at most $e (1 + e(\degree - 1))^{w-1}$ polymers $\gamma$ such that $a \in \gamma$ and $\abs{\gamma} = w$, and this is equal to the number of polymers $\gamma'$ with $\abs{\gamma'} = w - 1$ and which are adjacent to $a$ (provided $w > 1$).
So,
\begin{align*}
    \sum_{\gamma \in K} \abs{\wh{w}_\gamma} e^{\xi(\gamma)} \iver{\gamma \nsim \gamma_{\{a\}}}
    &\leq \sum_{w = 2}^{\infty} (e\beta_c)^{w-1} e(1 + e(\degree - 1))^{w - 1} \\
    &= e \sum_{w = 1}^{\infty} (e\beta_c(1 + e(\degree - 1)))^{w} \\
    &= e \sum_{w = 1}^{\infty} \frac{1}{(e+1)^w} = 1
\end{align*}
where the last line uses that $\beta_c = \frac{1}{e(e+1)(1 + e(\degree - 1))}$.
Consequently, if we take a general $\gamma_*$, then we get the desired bound:
\begin{equation*}
    \sum_{\gamma \in K} \abs{\wh{w}_\gamma} e^{\xi(\gamma)} \iver{\gamma \nsim \gamma_*}
    \leq \sum_{a \in \gamma_*} \sum_{\gamma \in K} \abs{\wh{w}_\gamma} e^{\xi(\gamma)} \iver{\gamma \nsim \gamma_{\{a\}}} \leq \abs{\gamma_*}. \qedhere
\end{equation*}
\end{proof}

\begin{corollary} \label{cor:cluster-truncation}
    When $\beta < \beta_c = \frac{1}{e(e+1)(1 + e(\degree - 1))}$, the cluster expansion of the polymer model in \cref{lem:mh2} can be truncated.
    Then the error term from truncation at degree $d$ is bounded,
    \begin{align*}
        \abs[\Big]{\sum_{\ell = d+1}^\infty \sum_{k = 1}^{\infty} \sum_{\substack{(\gamma_1,\dots,\gamma_k) \in K^{k} \\ \abs{\gamma_1} + \dots + \abs{\gamma_k} = \ell}} \varphi(H_{(\gamma_1,\dots,\gamma_k)}) \prod_{i \in [k]} w_{\gamma_i}}
        \leq \qubits \frac{(\beta/\beta_c)^{d+1}}{1 - \beta/\beta_c}.
    \end{align*}
\end{corollary}
\begin{proof}
This follows from both parts of \cref{lem:mh4}.
First, let $T \subseteq [\terms]$ be a set of terms such that $\abs{T} \leq \qubits$ and every term is a neighbor of some element of $T$ in $\graph$.
\begin{align*}
    & \abs[\Big]{\sum_{\ell = d+1}^\infty \sum_{k = 1}^{\infty} \sum_{\substack{(\gamma_1,\dots,\gamma_k) \in K^{k} \\ \abs{\gamma_1} + \dots + \abs{\gamma_k} = \ell}} \varphi(H_{(\gamma_1,\dots,\gamma_k)}) \prod_{i \in [k]} w_{\gamma_i}} \\
    &\leq \sum_{\ell = d+1}^\infty \sum_{k = 1}^{\infty} \sum_{\substack{(\gamma_1,\dots,\gamma_k) \in K^{k} \\ \abs{\gamma_1} + \dots + \abs{\gamma_k} = \ell}} \abs{\varphi(H_{(\gamma_1,\dots,\gamma_k)})} \prod_{i \in [k]} \abs{w_{\gamma_i}} \\
    &= \sum_{\ell = d+1}^\infty \sum_{\substack{\gamma_1 \in K \\ \abs{\gamma_1} \leq \ell}} \abs{w_{\gamma_1}} \parens[\big]{\frac{\beta_c}{\beta}}^{-\ell + \abs{\gamma_1}} \sum_{k = 1}^{\infty} \sum_{\substack{(\gamma_2,\dots,\gamma_k) \in K^{k-1} \\ \abs{\gamma_1} + \dots + \abs{\gamma_k} = \ell}} \abs{\varphi(H_{(\gamma_1,\dots,\gamma_k)})} \prod_{i = 2}^k \abs{w_{\gamma_i}}\parens[\big]{\frac{\beta_c}{\beta}}^{\abs{\gamma_i}} \\
    &\leq \sum_{\ell = d+1}^\infty \sum_{\substack{\gamma_1 \in K \\ \abs{\gamma_1} \leq \ell}} \abs{w_{\gamma_1}} \parens[\big]{\frac{\beta_c}{\beta}}^{-\ell + \abs{\gamma_1}} e^{\abs{\gamma_1}}
    \tag*{by \cref{eq:kp-conclusion}}\\
    &= \sum_{\ell = d+1}^\infty \parens[\big]{\frac{\beta_c}{\beta}}^{-\ell} \sum_{\substack{\gamma_1 \in K \\ \abs{\gamma_1} \leq \ell}} \abs{\wh{w}_{\gamma_1}} e^{\abs{\gamma_1}} \\
    &\leq \sum_{\ell = d+1}^\infty \parens[\big]{\frac{\beta_c}{\beta}}^{-\ell} \sum_{a \in T} \sum_{\substack{\gamma_1 \in K \\ \abs{\gamma_1} \leq \ell}} \abs{\wh{w}_{\gamma_1}} e^{\abs{\gamma_1}} \iver{\gamma_1 \nsim \gamma_{\{a\}}} \tag*{by property of $T$} \\
    &\leq \sum_{\ell = d+1}^\infty \parens[\big]{\frac{\beta_c}{\beta}}^{-\ell} \abs{T} \tag*{by \cref{eq:kp-assumption}}\\
    &= \abs{T}\frac{(\beta/\beta_c)^{d+1}}{1 - \beta/\beta_c}
    \leq \qubits\frac{(\beta/\beta_c)^{d+1}}{1 - \beta/\beta_c} \tag*{by property of $T$}
\end{align*}
\end{proof}

\subsection{Computation of the cluster expansion}

\begin{lemma}[{Version of \cite[Theorem 3.1]{hkt21} for the log-partition function}]
    \label{lem:cluster}
    Let $H = H(\lambda) = \sum_{a=1}^\terms \lambda_a E_a$ be a Hamiltonian on $\qubits$ qubits with Pauli terms, locality $\locality$, and dual interaction graph $\graph$ with max degree $\degree$ (\cref{def:hamiltonian,def:low-intersection-ham}).
    Let $0 \leq \beta < \beta_c = 1/(e(e+1)(1 + e(\degree - 1)))$.
    Then the log-partition function of $H$ can be expressed as a power series in $\beta$,
    \[
        \logpart \deq \log(\tr(e^{-\beta H})) = \sum_{\ell \geq 0} \beta^\ell p_\ell(\lambda),
    \]
    where $p_\ell$ is a degree-$\ell$ homogeneous polynomial in $\lambda$ with the following properties:
    \begin{enumerate}
        \item $p_\ell$ consists of at most $\qubits (e\degree)^\ell$ monomials.
        \item The series decays as $\abs{\sum_{\ell \geq d} \beta^\ell p_\ell(\lambda)} \leq \qubits\frac{(\beta/\beta_c)^d}{1 - \beta/\beta_c}$.
    \end{enumerate}
    Further, after $\bigO{\locality \terms \degree \log \degree}$ pre-processing time, we have the following form of access to $p_\ell$:
    \begin{enumerate}[label=\Alph*.]
        \item The list of monomials that appear in $p_\ell$ can be enumerated in time $\bigO{\ell\degree \mu}$, where $\mu$ is the number of monomials.
        \item The coefficient of any monomial in $p_\ell$ can be computed exactly in $\bigO{\locality \ell^3 + 8^\ell \ell^5\log^2\ell} = (8^\ell + \locality)\poly(\ell)$ time.
    \end{enumerate}
\end{lemma}
\begin{proof}
Most of this directly follows from the results in Section 3 of \cite{hkt21}, with property 2.\ following from \cref{cor:cluster-truncation}.
We direct the reader to \cite{hkt21} for the proofs of those results.
Here, we only identify where our statements come from.
 
First, we observe that the $\logpart = \log(\tr(e^{-\beta H}))$ has a formal multivariate Taylor series expansion around $\lambda = (0, \dots, 0)$ \cite[Eqs.\ 24 and 25]{hkt21},
\begin{align*}
    \logpart = \sum_{\ell \geq 0} \underbrace{\sum_{\cluster V : \abs{\cluster V} = \ell} \frac{\lambda^{\cluster V}}{\cluster V!} \mdiff_{\cluster V} \logpart}_{\beta^\ell p_\ell(\lambda)},
\end{align*}
where $\cluster V$ denotes a multiset over terms $[\terms]$; $\lambda^{\cluster V} = \prod_{a \in \cluster V} \lambda^a$ is the product of all coefficients associated to the terms in $\cluster V$ with multiplicity; and $\mdiff_{\cluster V}\logpart = \prod_{a \in \cluster V} \frac{\partial}{\partial \lambda_a} \logpart |_{\lambda = (0,\dots,0)}$ is the log-partition function, with derivatives taken for every $\lambda_a$ with $a \in \cluster V$ with multiplicity, evaluated at $\lambda = (0,\dots,0)$.
Note that $\mdiff_{\cluster V} \logpart$ is a constant in $\lambda$.
This formal expression becomes a true equality whenever the right-hand side series converges.

The coefficient, $\mdiff_{\cluster V} \logpart$, is only non-zero when $\cluster V$ is connected \cite[Proposition 3.5]{hkt21}.
Because of the degree bound $\degree$, the number of such ``clusters'' (connected $\cluster V$) of size $\ell$ is merely exponential, bounded by $\qubits e\degree (1 + e(\degree - 1))^{\ell - 1} \leq \qubits (e \degree)^\ell$ \cite[Proposition 3.6]{hkt21}.
This gives the monomial bound.

As discussed in the previous section, the cluster expansion of $\log(\ntr(e^{-\beta H})) = \log(\tr(e^{-\beta H})) - \log(2^\qubits)$ is equal to its multivariate Taylor series expansion when terms associated to a particular monomial are collected together.
In particular, we have that, in the notation from \cref{lem:mh2},
\begin{align*}
    \sum_{\ell \geq d+1} \beta^\ell p_\ell(\lambda)
    = \sum_{\ell \geq d+1} \sum_{k = 1}^{\infty} \sum_{\substack{(\gamma_1,\dots,\gamma_k) \in K^{k} \\ \abs{\gamma_1} + \dots + \abs{\gamma_k} = \ell}} \varphi(H_{(\gamma_1,\dots,\gamma_k)}) \prod_{i \in [k]} w_{\gamma_i}.
\end{align*}
Thus, we can apply the bound in \cref{cor:cluster-truncation} to get property 2.

As for the running time statements, enumerating monomials amounts to enumerating clusters, which is done in \cite[Section 3.4]{hkt21}.
Computing a coefficient amounts to computing $\mdiff_{\cluster V}\logpart$, which is done in \cite[Proposition 3.13]{hkt21}.
\end{proof}

\end{document}